\documentclass[lettersize,journal]{IEEEtran}
\usepackage{amsmath,amsfonts,amsthm,amssymb,mathrsfs}
\usepackage{algorithmic}
\usepackage{array}
\usepackage[caption=false,font=normalsize,labelfont=sf,textfont=sf]{subfig}
\usepackage{textcomp}
\usepackage{stfloats}
\usepackage{url}
\usepackage{verbatim}
\usepackage{graphicx}
\usepackage{bbding}
\usepackage{multirow}
\usepackage{cite, bm, bbm, xcolor}
\usepackage{refcount}
\usepackage{flafter}
\usepackage{makecell}
\usepackage{diagbox}
\usepackage{booktabs}
\usepackage{hyperref}
\theoremstyle{definition}
\newtheorem{theorem}{Theorem}
\newtheorem{proposition}{Proposition}
\newtheorem{lemma}{Lemma}
\newtheorem{remark}{Remark}
\theoremstyle{definition}
\newtheorem{definition}{Definition}
\usepackage{algorithm}

\usepackage{subcaption}
\usepackage{ulem}

\usepackage[T1]{fontenc}
\providecommand{\url}[1]{#1}
\allowdisplaybreaks[4]

\begin{document}
\title{Pinching-Antenna Systems (PASS): Power Radiation Model and Optimal Beamforming Design}

\author{
Xiaoxia Xu, Xidong Mu, Zhaolin Wang, 
Yuanwei Liu, \textit{Fellow, IEEE},  
and Arumugam Nallanathan, \textit{Fellow, IEEE}
\thanks{X. Xu, Z. Wang, and A. Nallanathan are with the School of Electronic Engineering and Computer Science, Queen Mary University of
London, London E1 4NS, U.K. (e-mail: \{x.xiaoxia, zhaolin.wang, a.nallanathan\}@qmul.ac.uk).}
\thanks{X. Mu is with the Centre for Wireless Innovation (CWI), Queen's University Belfast, Belfast, BT3 9DT, U.K. (e-mail: x.mu@qub.ac.uk)}
\thanks{Y. Liu is with the Department of Electrical and Electronic Engineering (EEE), The University of Hong Kong, Hong Kong (e-mail: yuanwei@hku.hk).}
\vspace{-3em}
}
\vspace{-1.5em}

\maketitle

\vspace{-1.5em}
\begin{abstract}  
    Pinching-antenna systems (PASS) improve wireless links by configuring the locations of activated pinching antennas along dielectric waveguides, namely pinching beamforming.
    In this paper, a novel adjustable power radiation model is proposed for PASS, 
    where power radiation ratios of pinching antennas can be flexibly controlled by tuning coupling spacing between pinching antennas and waveguides. 
    The closed-form coupling spacings are derived to achieve flexible and equal-power radiation. 
    Based on the commonly-assumed equal-power radiation, a practical PASS framework relying on discrete activation is considered, where pinching antennas can only be activated among a set of predefined locations. 
    A transmit power minimization problem is formulated, which jointly optimizes the transmit beamforming, pinching beamforming, and the numbers of activated pinching antennas, subject to each user's minimum rate requirement. 
    (1) To obtain globally optimal solutions of the resulting highly coupled mixed-integer nonlinear programming (MINLP) problem, 
    branch-and-bound (BnB)-based algorithms are proposed for both single-user and multi-user scenarios.
    (2) A low-complexity many-to-many matching algorithm is further developed. Combined with the Karush-Kuhn-Tucker (KKT) theory, locally optimal and pairwise-stable solutions are obtained within polynomial-time complexity. Simulation results demonstrate that: 
    (i) PASS significantly outperforms conventional multi-antenna architectures, particularly when the number of users and the spatial range increase; and 
    (ii) The proposed matching-based algorithm achieves near-optimal performance, resulting in only a slight performance loss while significantly reducing computational overheads.
    Code is available at \url{https://github.com/xiaoxiaxusummer/PASS_Discrete}. 
\end{abstract}
\begin{IEEEkeywords}
    Activation, beamforming, optimization, pinching antenna, pinching-antenna system (PASS).
\end{IEEEkeywords}

\section{Introduction}
Wireless networks have been long pursuing higher capacity and enhanced connectivity, 
driving the development of flexible-antenna techniques in the sixth-generation (6G) networks, such as reconfigurable intelligent surfaces (RISs) \cite{RIS}, 
simultaneous transmitting and reflecting surfaces (STARS) \cite{STARS}, 
and fluid/movable antennas \cite{FluidAntenna,MovableAntenna}. 
By manipulating the propagation environment or antenna geometry, 
existing flexible-antenna techniques significantly enhance spectral and energy efficiency, but primarily affect small-scale fading and local scattering characteristics.  
To further adjust large-scale channel effects, e.g., path loss and long-range shadowing, 
pinching-antenna system (PASS) has emerged recently as a revolutionary flexible-antenna technique \cite{PAr_Liu,PAr_Ding}. 
PASS enables wireless signals to  dynamically and closely follow mobile users. 
The original concept of PASS and the world's first prototype were developed by NTT DOCOMO \cite{PA_DOCOMO}.
Specifically, PASS comprise dielectric waveguides spanning across several to tens of meters. 
Acting similar to leaky-wave antennas, these waveguides transmit and receive wireless signals via multiple small dielectric particles, known as \textit{pinching antennas}, that are discretely attached along the waveguide \cite{PA_DOCOMO}. 
Since pinching antennas can be deployed and selectively activated at arbitrary positions along the waveguide, 
signal radiation and reception can be delivered to the ``last meter'' of user proximity. 
Therefore, PASS not only reshape the path loss profile experienced by mobile users but also maintain line-of-sight (LoS) connectivity, even in dense obstacle environments.

PASS reconfigure wireless links by changing the  locations of activated pinching antennas along waveguides, namely \textit{pinching beamforming}. 
Existing PASS structures can be categorized into continuous activation and discrete activation \cite{PAr_Liu}. 
Continuous activation allows pinching antennas to be placed/activated at arbitrary positions over waveguides. 
In contrast, discrete activation selectively activates pinching antennas installed at a finite set of pre-configured discrete locations, thus reducing 
the implementation complexity and hardware complexity. 
For single-waveguide continuous activation, the authors of \cite{Rate_DL_PAr_SU} maximized the downlink transmission rate by optimizing the locations of the pinching antennas. 
Moreover, the authors of \cite{Rate_UL_Par_MU} jointly optimized the positions of pinching antennas and the bandwidth/time resource allocation for uplink communications.
The authors of \cite{Optimal_spacing} derived the optimal number of pinching antennas and the optimal inter-antenna spacing to maximize the array gain. 
The authors of  \cite{outage_PASS} derived the outage probability and average rate, and analyzed the optimal placement for a single pinching antenna. 
The authors of \cite{Discrete_matching} investigated the discrete activation of pinching antennas for non-orthogonal multiple access (NOMA) assisted PASS. 
Considering a general multi-waveguide PASS in downlink multiple-input single-output (MISO) networks, 
the authors of \cite{PASS_modeling} revisited the physics of PASS based on electromagnetic coupling and derived power radiation models. 
Penalty-based and zero-forcing based beamforming designs were developed for both continuous and discrete activations. 
For continuous activation, the authors of \cite{PASS_beamforming_optimization} proposed both optimization-based and learning-based methods 
to jointly optimize transmit and pinching beamforming for system sum rate maximization. 
More recently, \cite{PASS_ISAC} investigated PASS-enabled integrated sensing and communications by exploiting look-angle dependent radar cross-section (RCS) to achieve target diversity, 
and \cite{PASS_near_far} introduced wireless powered pinching antenna networks to address the double near-far problem.

While existing studies have demonstrated the promising prospects of PASS, two fundamental challenges remain unresolved for unlocking its potentials:
\textbf{(i) Globally optimal PASS design:} 
The joint optimization of conventional digital beamforming and pinching beamforming suffers from a highly coupled and nonconvex problem structure. 
Hence, existing studies typically investigated 
suboptimal joint beamforming solutions under a fixed number of activated antennas \cite{PASS_modeling,PASS_beamforming_optimization}. 
However, the globally optimal joint beamforming design of PASS remains unexplored, and the performance gaps between suboptimal solutions and the global optimum are unknown.
\textbf{(ii) Adjustable power radiation}: 
    How to achieve adjustable power radiation control remains an open problem. 
    Existing PASS studies commonly relied on the equal-power radiation assumption \cite{Rate_DL_PAr_SU,PASS_beamforming_optimization}. 
    The fundamental physical power radiation model was proposed in \cite{PASS_modeling} based on coupled-mode theory.  
    However, the power radiation ratio is determined by customizing coupling length of each pinching antenna. 
    Since the coupling length is typically fixed by fabrication, it cannot be altered in the real time. 
    Even when the number of activated antennas change, all antennas' coupling lengths need to be reshaped for adaptation. Hence, adaptive designs are necessitated for dynamic adjustment. 

Against the above background, this paper proposes a novel adjustable power radiation model and a globally optimal beamforming design for PASS with discrete activation.
\textit{First}, the proposed power radiation model adjusts power radiation ratios by flexibly altering coupling spacing between pinching antennas and waveguides. 
We derive the closed-form coupling spacing adjustment scheme to ensure flexible or equal-power radiation given arbitrary combinations of activated pinching antennas. 
\textit{Secondly}, for the equal-power radiation case, we formulate the joint pinching beamforming and transmit beamforming optimization problem, 
 which is a nonconvex mixed integer nonlinear programming (MINLP). 
We propose globally optimal solutions based on a tailored branch-and-bound (BnB) approach.
\textit{Thirdly}, to achieve a low-complexity design, we develop a welfare-driven many-to-many matching algorithm, which is demonstrated to achieve near-optimal performance.
The key contributions of this paper are summarized as follows.
\begin{itemize}
    \item 
    We propose a novel adjustable power radiation model for PASS, where power radiation ratios can be controlled by tuning coupling spacing between waveguides and pinching antennas.
    We derive closed-form spacing solutions to support both flexible and equal power radiations given arbitrary antenna activation numbers and combinations. 
    Based on the generally assumed equal-power radiation, we investigate a practical PASS communication framework with discrete activation, where 
    pinching antennas can be activated from a set of pre-mounted discrete locations on waveguides.
    A transmit power minimization problem is formulated, which jointly optimizes the transmit beamforming, pinching beamforming, and the numbers of activated pinching antennas,
    while ensuring the minimum rate requirements of users.
   \item We propose globally optimal joint beamforming algorithms for both single-user and multi-user scenarios. 
    For the single-user scenario, the MINLP is reduced to a non-convex quadratic constrained quadratic programming (QCQP), and we develop a BnB algorithm to find the global optimum. 
    For the multi-user scenario, we construct tractable convex relaxation based on McCormick envelopes to enable bound estimation of BnB.
    We mathematically prove that the resulting algorithm can optimally determine the numbers and locations of activated antennas and the corresponding transmit beamforming.
    \item We further propose a low-complexity suboptimal algorithm based on the many-to-many matching theory. 
    The pinching antenna activation is modelled as a many-to-many matching game with externalities and non-substitutability,  
    where agents' preferences depend on beamforming solutions obtained via the Karush-Kuhn-Tucker (KKT) theory.
    A  welfare-driven many-to-many matching algorithm is developed,  
    which can converge to local optima in polynomial time complexity and ensure pairwise equilibrium.
    \item We provide numerical results to verify the effectiveness of the proposed algorithms, which demonstrate that:  
    Relying on the proposed framework, PASS outperforms conventional MIMO systems in both single-user and multi-user cases, and the achievable performance gains significantly increase with the number of multiplexed users and the spatial range. 
    Moreover, the proposed low-complexity many-to-many matching algorithm achieves near-optimal performance, which only suffers from a marginal loss compared to the optimal algorithm.
\end{itemize}

The rest of this paper is organized as follows. 
Section II presents the proposed PASS framework with adjustable power radiation model and formulates the optimization problem. 
Section III proposes the globally optimal BnB algorithms, and Section IV developes the low-complexity many-to-many matching algorithm. 
Section V provides numerical results to verify the efficiency of the proposed framework and algorithms. 
Finally, Section VI concludes the paper.

\textit{Notations}: The variable, vector, and matrix are denoted by $x$, $\mathbf{x}$, and $\mathbf{X}$, respectively. $|x|$ denotes the absolute value of a real number and the modulus of a complex number. 
$\left\Vert\mathbf{x}\right\Vert$ is the vector Euclidean norm, 
and $\left\Vert\mathbf{X}\right\Vert$ is the matrix Frobenius norm.
$\mathrm{Re}\left\{x\right\}$ and $\mathrm{Im}\left\{x\right\}$ denote the real and image parts of $x$, and $x^{H}$ is the complex conjugate number of $x$. 
$\mathbf{1}_{N\times 1}$ denotes an $N$-dimension all-ones vector.
$\mathbf{X}^{T}$ and $\mathbf{X}^{H}$  denote the transpose and the Hermitian matrix.

\begin{figure}[!t]
    \centering
    \includegraphics[width=0.49\textwidth]{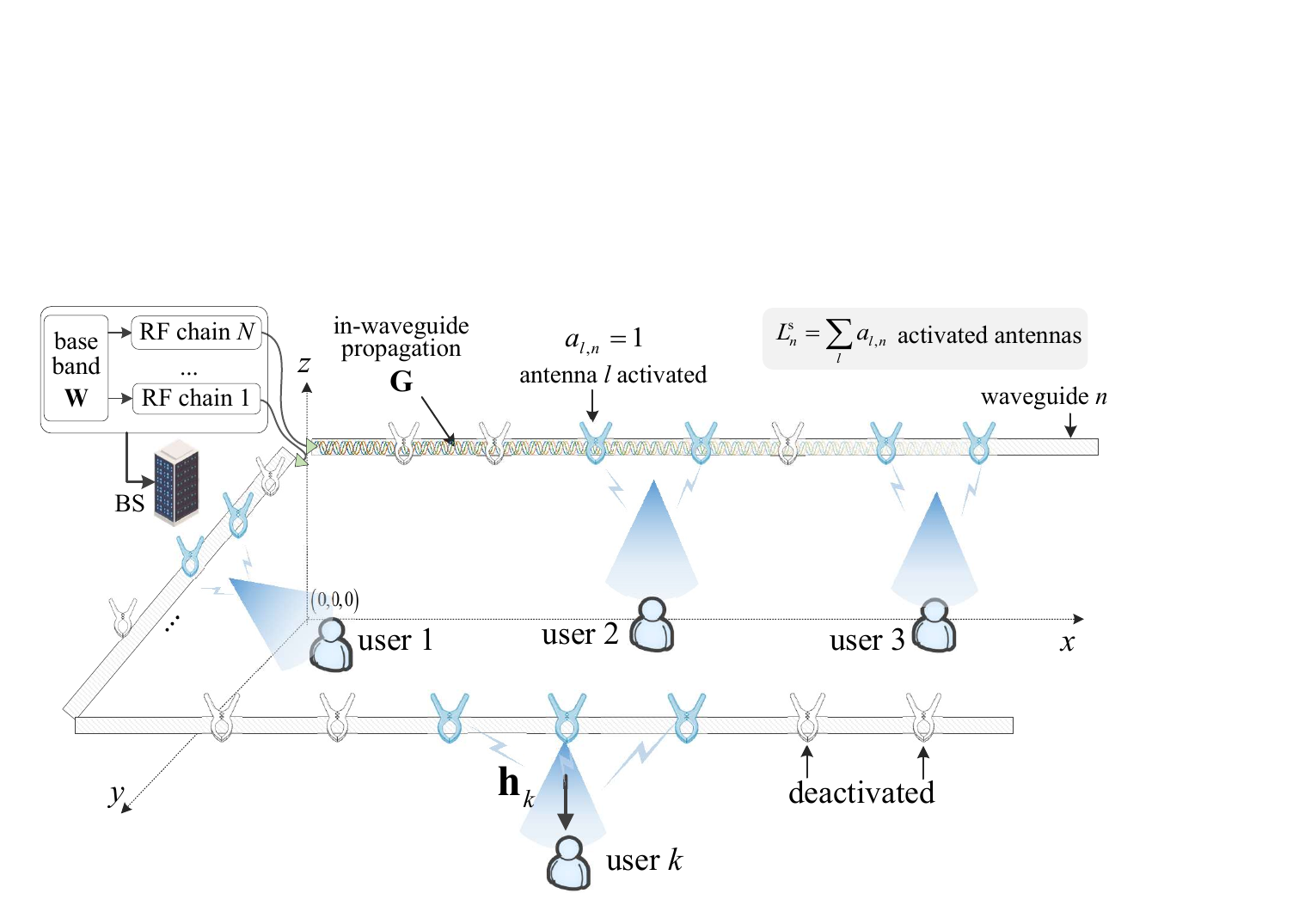}
    \caption{PASS enabled downlink MISO communication with discrete pinching antenna activation.}\label{fig_pass}
\end{figure}

\section{System Model and Problem Formulation} 

As shown in Fig. \ref{fig_pass}, we consider a PASS enabled downlink MISO communication framework with discrete activation, which serves a set $\mathcal{K}$ of $K$ single-antenna users.   
The PASS comprise a set $\mathcal{N}$ of $N$ waveguides  deployed over a rectangular area of $S_{\mathrm{x}} \times S_{\mathrm{y}}$ $\text{m}^2$. 
$L$ pinching antennas, indexed by $\mathcal{L}=\{1,2,\dots,L\}$, are pre-mounted at a set of discrete locations over each waveguide. 
The total number of pinching antennas is $M = NL$. 
Each waveguide is connected to a single radio frequency (RF) chain, thus enabling baseband processing and spatial multiplexing. 
The multiplexed baseband signals are fed into waveguides and radiated via the activated pinching antennas.
By selectively activating the pinching antennas, the system can adjust both the phases and the large-scale path loss of incident signals, 
leading to low-cost pinching beamforming design. 
We adopt a three-dimensional (3D) Cartesian coordinate. 
The base station (BS) of PASS is deployed at $\bm{\eta}_{0}=(0,0,h_{\mathrm{PA}})$, 
where $h^{\mathrm{PA}}$ is the fixed height. 
The waveguides can be deployed parallel to both the $x$- and $y$-axes, thus accommodating different user distributions. 
Moreover, the feed point of each waveguide is fixed at $\bm{\eta}_{n}^{\mathrm{W}} = \left[x_{n}^{\mathrm{W}}, y_{n}^{\mathrm{W}},  h^{\mathrm{PA}}\right]^{T}$,  
and $\bm{\eta}_{k}^{\mathrm{U}} = \left[x_{k}^{\mathrm{U}}, y_{k}^{\mathrm{U}}, 0\right]^{T}$ denotes the position of user $k$.
Each pinching antenna $l$ on waveguide $n$ is mounted at a fixed location $\bm{\eta}_{l,n} = \left[x_{l,n}^{\mathrm{PA}}, y_{l,n}^{\mathrm{PA}}, h^{\mathrm{PA}}\right]^{T}$ from a discrete set,  
where $x_{l,n}^{\mathrm{PA}}$ and $y_{n}^{\mathrm{W}}$ are pre-defined coordinates.
Let $\mathbf{a}_{n}=\left[a_{1,n}, a_{2,n}, \dots, a_{L,n}\right]^{T}\in\mathbb{Z}^{L\times 1}$ denote the binary pinching antenna activation vector for waveguide $n$, 
where $a_{l,n}=1$ if pinching antenna $l$ at waveguide $n$ is activated, and  $a_{l,n}=0$ otherwise.

\begin{figure}[!t]
    \centering
    \includegraphics[width=0.49\textwidth]{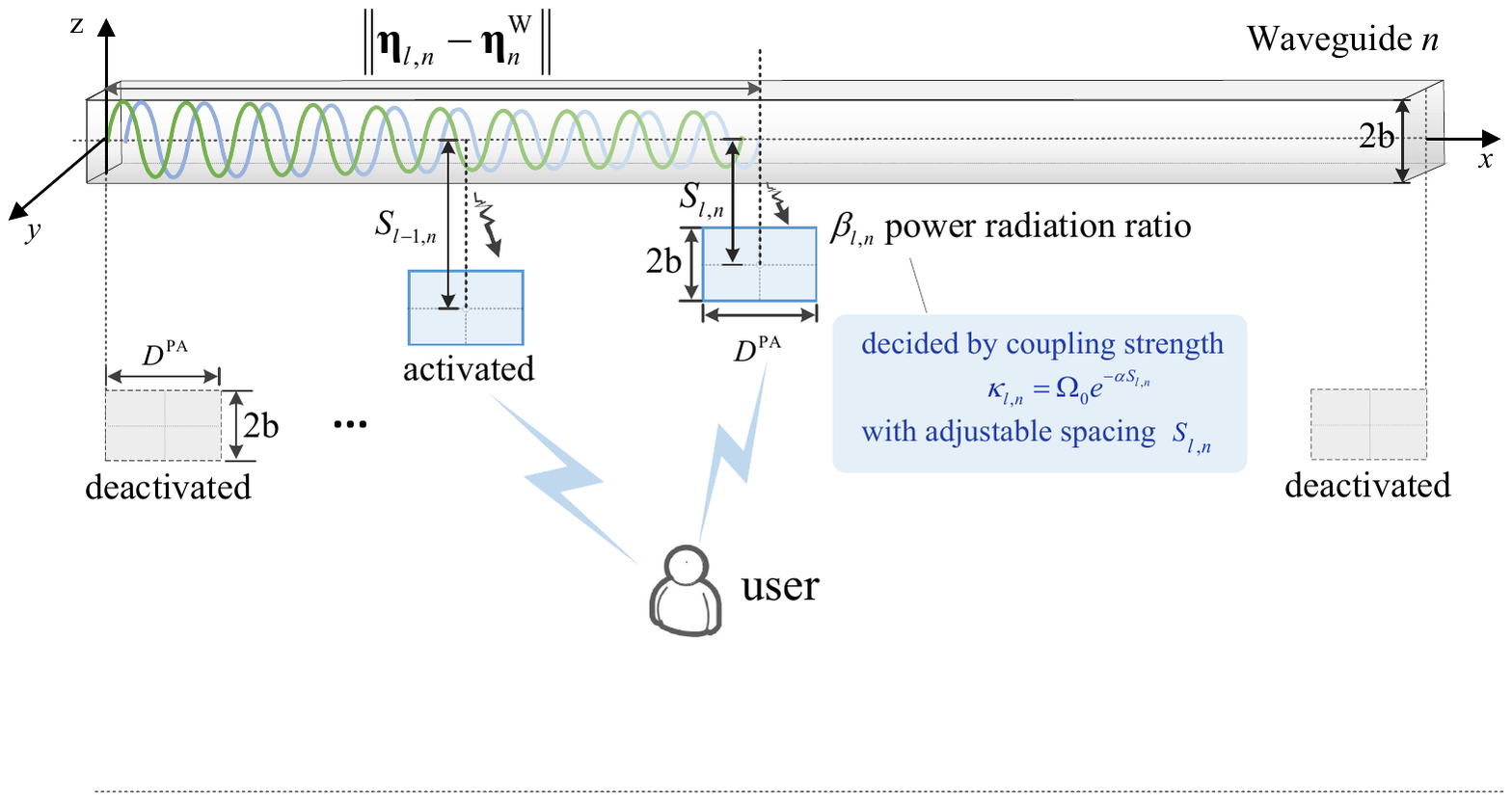}
    \caption{The proposed adjustable power radiation model with a local Cartesian coordinate system. The power radiation ratio of each pinching antenna is decided by adjustable spacing $S_{l,n}$.}
    \label{fig_waveguide}
\end{figure}

\subsection{Adjustable Power Radiation Model}
The power exchange between the waveguide and the adjacent pinching antennas can be modelled by coupled-mode theory (CMT)
under weak coupling and single-mode assumptions \cite{coupled_mode_theory}. 
By extending the analysis in \cite{PASS_modeling}, the power radiation ratio of pinching antenna $l$ at waveguide $n$ can be determined by  
the number and the order of activated antennas at this waveguide:
\begin{equation}\label{power_radiating_ratio}
    \beta_{l,n} \!=\! a_{l,n}\sin\left(\kappa_{l,n} D^{\mathrm{PA}}\right) \prod_{i=1}^{l-1} \sqrt{1-a_{i,n}\sin^{2}\left(\kappa_{i,n} D^{\mathrm{PA}}\right)},
\end{equation}
where $D^{\mathrm{PA}}$ is the fixed fabricated length of each pinching antenna, 
and coupling coefficient $\kappa_{l,n}$ measures power exchange strength (i.e., the coupling strength) from waveguide $n$ to pinching antenna $l$. 

To enable adjustable power radiation given a fixed fabricated length $D^{\mathrm{PA}}$, 
we propose a novel power radiation model that flexibly adjusts coupling coefficient $\kappa_{l,n}$ in \eqref{power_radiating_ratio}. 
This is achieved by tuning the coupling spacing $S_{l,n}$ between pinching antenna $l$ and waveguide $n$, as shown in Fig. \ref{fig_waveguide}. 
Note that a pinching antenna can be regarded as a non-contact coupler (a small tap) over the waveguide. 
Coupling spacing $S_{l,n}$ determines power radiation behaviors by changing the power exchange ratio between the waveguide and pinching antenna, as modelled as follows. 
From CMT, $\kappa_{l,n}$ is given by the overlap integrals of the mode fields \cite{coupled_mode_theory}:
\begin{equation}\label{CMT_intgeral}
\kappa_{l,n} = \frac{\omega\epsilon_{0}}{4}  \iiint_{V(S)}\!\Delta \epsilon ~\mathbf E_{{\rm wg}, n}\!\cdot\!\mathbf E_{{\rm pa},l}^{*},\mathrm dV(S),
\end{equation}
where $\mathbf E_{{\rm wg},n}$ and $\mathbf E_{{\rm pa},l}$ are power-normalized modal electric field distributions of waveguide $n$ and pinching antenna $l$, respectively; 
$\epsilon_{0}$ is the vacuum permittivity; $\Delta \epsilon$ denotes coupling dielectric perturbation; 
$V(S)$ is the coupling region determined by $S$, 
and $\omega$ is the angular frequency. 
Built on the analytical derivations in \cite{Okamoto_waveguide}, we explicitly model the relationship between coupling coefficient $\kappa_{l,n}$ in \eqref{CMT_intgeral} 
and waveguide-antenna spacing $S_{l,n}$ by the following proposition. 
\begin{proposition}\label{proposition:coupling_strength}
    The coupling coefficient $\kappa_{l,n}$ in \eqref{CMT_intgeral} can be modelled as a function of
    coupling spacing $S_{l,n}$ as follows\footnote{As proven in Appendix \ref{proof:proposition:coupling_strength}, \eqref{eq:coupling_coefficient} provides a tractable approximation, 
    where $\Omega_{0}$ reflects geometry/material factors and $\alpha$ captures guided-mode dispersion. 
    More complex effects (e.g., multi-mode dispersion)  are left for future work.}:
    \begin{equation}\label{eq:coupling_coefficient}
        \kappa_{l,n} = \Omega_{0} e^{-\alpha S_{l,n}},
     \end{equation}
    where coefficient $\Omega_{0}$ captures electric distribution and modal normalization, 
    and the cladding decay constant  $\alpha = \sqrt{\gamma_{0}^{2} -\frac{4\pi^2}{\lambda_{f}^2}n_{\mathrm{clad}}^{2}}$ is determined by pinching antennas' propagation constant $\gamma_{0}$ and cladding refraction index $n_{\mathrm{clad}}$.
    \begin{proof}
        See Appendix \ref{proof:proposition:coupling_strength}.
    \end{proof}
\end{proposition}

Based on \textbf{Proposition \ref{proposition:coupling_strength}}, we can derive the following adjustable power radiation model. 

\begin{lemma}[Element-Wise Adjustable Power Radiation Model]\label{lemma:adjustable_power}
    The target power radiation ratio $\beta_{l,n}^{\mathrm{target}}$ can be achieved by one-by-one adjusting spacing $S_{l,n}$ between each activated pinching antenna $l$ and waveguide $n$, such that
    \begin{equation}\label{adjustable_power_model}
        \sin\!\left(\Omega_{0}e^{-\alpha S_{l,n}}D^{\mathrm{PA}}\right)\!=\!\frac{\beta_{l,n}^{\mathrm{target}}}{\prod_{i=1}^{l-1}\left(1\!-\!a_{i,n}\sin^{2}\!\left(\Omega_{0}e^{-\alpha S_{i,n}}D^{\mathrm{PA}}\right)\right)},  
    \end{equation}
    where $S_{i,n}$, $i = 1,2,\dots, l-1$, is determined before $S_{l,n}$.
\end{lemma}

Denote the number of activated pinching antennas on waveguide $n$ by $L_n^{\mathrm{s}} = \sum_{l=1}^L a_{l,n}$.
Using the proposed power radiation model \eqref{adjustable_power_model}, we introduce the following lemma to realize the commonly assumed equal-power radiation \cite{PAr_Ding, PASS_beamforming_optimization}, 
where the power radiation ratios of antennas on each waveguide dynamically change with combinations of $\mathbf{a}_{n}$:
\begin{equation}\label{equal_power_radiation}
    \beta_{l,n} = \beta_{n} = \frac{1}{\sqrt{L_n^{\mathrm{s}}}}=\frac{1}{\sqrt{\sum_{l=1}^{L}a_{l,n}}}, \quad \forall l \in \mathcal{L}, ~ n \in \mathcal{N}. 
\end{equation}
\begin{lemma}[Special Case: Equal-Power Radiation]\label{lemma:equalpower} 
    To achieve equal power radiation \eqref{equal_power_radiation}, 
    the coupling spacing $S_{l,n}$, $l\!=\!1,\dots,L$, can be sequentially adjusted over waveguide $n$ by
    \vspace{-0.5em}
    \begin{equation}\label{equal_power_spacing}
            S_{l,n}\!=\!\frac{1}{\alpha}\,\ln\!\left(\frac{\Omega_{0}D^{\mathrm{PA}}}{\arcsin(\delta_{l,n})}\right), ~ 
            \delta_{l,n} \!=\! \frac{1}{\sqrt{L_{n}^{\mathrm{s}} \!-\! \rho_{l,n}}}, ~ 
            \text{if $a_{l,n}=1$},
    \end{equation}
    where $\rho_{l,n} = \sum_{i=1}^{l-1} a_{i,n}$ denotes the number of activated pinching antennas that are deployed closer to the feed point 
    and would radiate power prior to pinching antenna $l$.
    \begin{proof}
        See Appendix \ref{proof:lemma:equalpower}.
    \end{proof}
\end{lemma}

Without loss of generality, this paper investigates globally optimal and near-optimal designs under the equal-power radiation scheme, which is more common and easy to implement.  
Note that non-equal power radiation optimization grounded in Lemma \ref{lemma:adjustable_power} is also an important direction. 
Due to space limits, we leave it for future research.

\subsection{PASS Signal Model}

\subsubsection{Signal Radiation within Waveguides}
The data signals of $K$ users are multiplexed at the baseband using the digital transmit beamforming matrix 
$\mathbf{W}=\left[\mathbf{w}_{1},\mathbf{w}_{2},\dots, \mathbf{w}_{K}\right]\in\mathbb{C}^{N\times K}$, 
where $\mathbf{w}_{k}$ denotes the transmit beamforming vector for the signal of user $k$.  
The baseband-multiplexed signal is modulated and passed through the RF chain, and then fed into the according waveguide. 
Hence, the transmitted signal $\mathbf{s}_{k,n} \in \mathbb{C}^{L \times 1}$ of user $k$ passed on the pinching antennas at waveguide $n$ is given by
\begin{equation}\label{inwaveguide_radiation}
    \mathbf{s}_{k,n}
        = \mathrm{diag}\left(\mathbf{g}_{n}\right) \mathbf{a}_{n} w_{n,k}\widetilde{s}_{k}, 
\end{equation}
where the baseband signal $\widetilde{s}_{k}$ satisfies $\mathbb{E}\left[\widetilde{s}_{k}^{H} \widetilde{s}_{k}\right] = 1$. 
$\mathbf{g}_{n} = \big[g_{1,n},g_{2,n},\dots,g_{L,n}\big]^{T}\in \mathbb{C}^{L\times 1}$ 
reflects the effects when signals propagate from the feed point of waveguide $n$ to the pinching antennas,   
which is given by\footnote{
We assume ideal waveguide propagation. 
In practice, dielectric waveguides exhibit frequency-dependent attenuation \cite{Loss_Waveguide}, which reduces the effective radiated power. 
The full impact can be investigated in the future work.}
\begin{equation}
g_{l,n} = \beta_{l,n} \widetilde{g}_{l,n} = \beta_{l,n} e^{-i\frac{2\pi}{\lambda_{\mathrm{w}}}\left\Vert \bm{\eta}_{l,n} \!-\! \bm{\eta}_{n}^{\mathrm{W}}\right\Vert}, 
\end{equation}
where $\beta_{l,n}$ denotes the power radiation ratio of pinching antenna $l$ at waveguide $n$,  
and $\widetilde{g}_{l,n}= e^{-i\frac{2\pi}{\lambda_{\mathrm{w}}}\left\Vert \bm{\eta}_{l,n} \!-\! \bm{\eta}_{n}^{\mathrm{W}}\right\Vert}$ 
is the in-waveguide response of the propagated signal. 
In addition, $\lambda_{\mathrm{w}}=\lambda_{f}/n_{\mathrm{eff}}$ denotes the guided wavelength, $\lambda_{f}$ indicates the wavelength of the carrier frequency,
and $n_{\mathrm{eff}}$ is the effective refractive index of the dielectric waveguide \cite{Waveguide,PAr_Ding}. 
$\left\Vert \bm{\eta}_{l,n}-\bm{\eta}_{n}^{\mathrm{W}}\right\Vert$ 
denotes the distance from the feed point of waveguide $n$ to pinching antenna $l$ at this waveguide.

The emitted signal $\mathbf{s}_{k} = \left[\mathbf{s}_{k,1}^{T},\mathbf{s}_{k,2}^{T}, \dots, \mathbf{s}_{k,N}^{T}\right]^{T} \in \mathbb{C}^{M\times 1}$ of pinching antennas from \eqref{inwaveguide_radiation} can be compactly written as
\begin{equation}
    \mathbf{s}_{k} = \mathbf{G}\mathbf{A}\mathbf{w}_{k}\widetilde{s}_{k},
\end{equation}
where the in-waveguide transmission response $\mathbf{G} \in \mathbb{C}^{M \times M}$ is defined as the diagonal matrix
\begin{equation}
        \mathbf{G}= \mathrm{diag} \left(\mathbf{g}_{1},\mathbf{g}_{2}, \dots, \mathbf{g}_{N}\right).
\end{equation}
Moreover, the discrete pinching antenna activation $\mathbf{A}\in\mathbb{Z}^{M\times N}$ is a block diagonal matrix, which is given by 
\begin{equation}
    \begin{split}
    \mathbf{A} &= \mathrm{blkdiag}\left(\mathbf{a}_{1}, \mathbf{a}_{2}, \dots, \mathbf{a}_{N}\right) 
    \\&=
        \left[\begin{array}{cccc}
            \mathbf{a}_{1} & \mathbf{0}_{L\times 1} & \dots & \mathbf{0}_{L\times 1}\\
            \mathbf{0}_{L\times 1} & \mathbf{a}_{2} & \dots & \mathbf{0}_{L\times 1}\\
            \vdots & \vdots & \ddots & \vdots\\
            \mathbf{0}_{L\times 1} & \mathbf{0}_{L\times 1} & \dots & \mathbf{a}_{N}
            \end{array}\right].
    \end{split}
\end{equation}

\subsubsection{Signal Radiation in Free Space}
We consider LoS-dominant channels in this work.
The downlink channel from pinching antennas $\mathcal{L}$ at waveguide $n$ to user $k$ is denoted by vector $\mathbf{h}_{n,k}^{H} \in\mathbb{C}^{1\times L}$. 
Based on the geometric free-space spherical wavefront model \cite{Spherical_Channel}, 
the channel coefficient from the $l$-th pinching antenna at waveguide $n$ to user $k$ located at position $\bm{\eta}_{k}^{\mathrm{U}}$ 
can be given by
\begin{equation}\label{channel_ori}
        h_{l,n,k}^{H} = 
        \frac{\sqrt{\varphi}e^{-i 2\pi/\lambda \left\Vert \bm{\eta}_{k}^{\mathrm{U}}-\bm{\eta}_{l,n}\right\Vert}}
        {\left\Vert \bm{\eta}_{k}^{\mathrm{U}}-\bm{\eta}_{l,n}\right\Vert},
\end{equation}
where 
$\lambda_{f}$ is the wavelength, 
and $\varphi ={c}/{\left(4\pi f_{c}\right)}$ denotes the reference channel gain depending on the speed of light $c$ and the carrier frequency $f_{c}$. 
Moreover, 
$\left\Vert \bm{\eta}_{k}^{\mathrm{U}}-\bm{\eta}_{l,n}\right\Vert$ 
is the distance between pinching antenna $l$ and user $k$, which is computed by 
$\left\Vert \bm{\eta}_{k}^{\mathrm{U}}\!-\!\bm{\eta}_{l,n}\right\Vert\!=\!
    \sqrt{\left(x_{l,n}^{\mathrm{PA}}\!-\!x_{k}^{\mathrm{U}}\right)^{2}
    \!+\!\left(y_{l,n}^{\mathrm{PA}}\!-\!y_{k}^{\mathrm{U}}\right)^{2} \!+\! \left(h^{\mathrm{PA}}\right)^{2}}$.

The channel vectors from pinching antennas to user $k$ are collected by $\mathbf{h}_{k}^{H} \!=\! \left[\mathbf{h}_{1,k}^{H}, \mathbf{h}_{2,k}^{H}, 
\dots, \mathbf{h}_{N,k}^{H}\right] \in \mathbb{C}^{1\times M}$. 
Hence, the received signal at user $k$ can be compactly expressed by
\begin{equation}
    y_{k}\!=\!{\underset{\text{pinching BF}}
    {\underbrace{\mathbf{h}_{k}^{H}\mathbf{G}\mathbf{A}}}
    \!\mathbf{w}_{k}\widetilde{s}_{k}}\!+\sum_{k'\ne k}
    {\!\underset{\text{pinching BF}}
    {\underbrace{\mathbf{h}_{k}^{H}\!\mathbf{G}\mathbf{A}}}
    \!\mathbf{w}_{k'}\widetilde{s}_{k'}}
    \!+\!n_{k}.
\end{equation} 
Therefore, the signal-to-interference-and-noise ratio (SINR) of user $k$ is given by
\begin{equation}
    \mathrm{SINR}_{k}\!=\!\frac{\left|\mathbf{h}_{k}^{H}\mathbf{G}\mathbf{A}\mathbf{w}_{k}\right|^{2}}
    {\sum\limits _{k'\ne k}\left|\mathbf{h}_{k}^{H}\mathbf{G}\mathbf{A}\mathbf{w}_{k'}\right|^{2}+\sigma^{2}}.
\end{equation} 

\subsection{Problem Formulation}
By determining transmit beamforming $\mathbf{W}$ and discrete activation $\mathbf{A}$, 
we jointly optimize the transmit beamforming, pinching beamforming, and the numbers of activated antennas $L_{n}^{\mathrm{s}}$ of all waveguides. 
The key goal is to minimize the transmit power subject to each user's SINR requirement, which is a classical design objective in wireless system design:
\begin{subequations}\label{problem_P}
    \begin{align*}
        \text{(P0)} ~ &  \min_{\mathbf{A}, \mathbf{W}}
        ~ \left\Vert \mathbf{W} \right\Vert_{F}^{2}, \tag{\ref{problem_P}{a}} \\
        \text{s.t.} ~
        & a_{l,n} \in \{0,1\},  ~ \forall l \in \mathcal{L}, ~ n\in\mathcal{N}, \tag{\ref{problem_P}{b}} \label{P_act} \\   
        & \frac{\left|\mathbf{h}_{k}^{H}\mathbf{G}\mathbf{A}\mathbf{w}_{k}\right|^{2}}
        {\sum\limits_{k'\ne k}\left|\mathbf{h}_{k}^{H}\mathbf{G}\mathbf{A}\mathbf{w}_{k'}\right|^{2}+\sigma^{2}} \geqslant \gamma^{\min}, 
        ~\forall k \in\mathcal{K}, \tag{\ref{problem_P}{c}} \label{P_rate}
    \end{align*}
\end{subequations}
where \eqref{P_act} is the binary constraint of the pinching antenna activation decision variable $a_{l,n}$,  
and \eqref{P_rate} ensures the minimum data SINR requirement of each user. 

Problem (P0) is a nonconvex MINLP with highly coupled variables. 
Specifically, the term $\mathbf{G}\mathbf{A}\mathbf{w}_{k}$ can be expressed as 
\begin{equation}\label{coupled_term_MU}
    \mathbf{G} \mathbf{A}  \mathbf{w}_{k} \!=\! \mathbf{\widetilde{G}} \mathbf{A} 
    \mathrm{diag}\!\left(\!\frac{1}{\sqrt{L_{1}^{\mathrm{s}}}},
    \frac{1}{\sqrt{L_{2}^{\mathrm{s}}}},\dots,\frac{1}{\sqrt{L_{N}^{\mathrm{s}}}}\right) \mathbf{w}_{k},
\end{equation}
where $L_{n}^{\mathrm{s}}=\sum_{l=1}^{L}a_{n,l}$ denotes the number of activated pinching antennas. 
Due to the dependence of $L_{n}^{\mathrm{s}}$ on the pinching antenna activation matrix $\mathbf{A}$,
\eqref{coupled_term_MU} suffers from strong coupling among the pinching antenna activation $\mathbf{A}$, 
the corresponding number of activated pinching antennas $L_{n}^{\mathrm{s}}$, and the transmit beamforming matrix $\mathbf{W}$.
Hence, the nonconvexity arises from both the strong variable coupling and the discrete structure of $\mathbf{A}$.
To search for the globally optimal solutions, 
we will first construct the convex relaxation for problem (P0) and come up with BnB-based globally optimal algorithms in the sequel, 
and then explore near-optimal and low-complexity design. 

\section{BnB-Based Globally Optimal Beamforming}
In this section, we propose BnB-based globally optimal beamforming algorithms for both single-user and multi-user scenarios in PASS. 

\subsection{Optimal Solution for Single-User Scenario}
For single-user scenario, problem (P0) defined in \eqref{problem_P} is reduced to the following form:
\begin{subequations}\label{P1}
    \begin{align*}
        \text{(P1)} ~ \min_{\mathbf{A}, \mathbf{w}}~ 
        &  P = \left\|\mathbf{w}\right\|^{2}, \tag{\ref{P1}{a}} \label{P1_obj}  \\
        \text{s.t.} ~
        & a_{l} \in \{0,1\},  ~ \forall l \in \mathcal{L}, \forall n\in\mathcal{N}, \tag{\ref{P1}{b}} \label{P1_act} \\   
        & \frac{\left|\mathbf{h}^{H}\mathbf{G}\mathbf{A}\mathbf{w}\right|^{2}}{\sigma^{2}} \geqslant \gamma^{\min}, 
        \tag{\ref{P1}{c}} \label{P1_rate}
    \end{align*}
\end{subequations}
where  $\mathbf{h} = \big[\mathbf{h}_{1}, \mathbf{h}_{2}, \dots, \mathbf{h}_{N} \big]^{T}\!\in\!\mathbb{C}^{M\times 1}$ is the downlink PASS channel to the user, and 
$\mathbf{h}_{n}^{H} = \big[h_{l,n}^{H}\big]=\bigg[\frac{\sqrt{\varphi}e^{-i2\pi/\lambda\left\Vert \bm{\eta}^{\mathrm{U}}-\bm{\eta}_{l,n}\right\Vert }} 
{\left\Vert \bm{\eta}_{0}^{\mathrm{U}}-\bm{\eta}_{l,n}\right\Vert }\bigg]$. 
Note that problem (P1) is an NP-hard MINLP. 
Furthermore, denote $\widetilde{\mathbf{g}}_{n} = \big[\widetilde{g}_{1,n},\widetilde{g}_{2,n},\dots,\widetilde{g}_{L,n}\big]^{T}$. 
Since $\mathbf{g}_{n}=\beta_{n}\widetilde{\mathbf{g}}_{n}=\frac{1}{\sqrt{\sum_{l=1}^{L}a_{l,n}}}\widetilde{\mathbf{g}}_{n}$ 
is a fractional function of $\mathbf{a}_{n}$, the term $\mathbf{G}\mathbf{A}\mathbf{w}$ is strongly coupled and nonconvex.

Fortunately, the optimal transmit beamforming $\mathbf{w}^{*}$ of (P1) is given by the maximum ratio transmission (MRT) strategy:
\begin{equation}\label{BF_SU}
    \mathbf{w}^{*}=\sqrt{P}\frac{\left(\mathbf{h}^{H}{\mathbf{G}}\mathbf{A}\right)^{H}}{\left\|\mathbf{h}^{H}{\mathbf{G}}\mathbf{A}\right\|}.
\end{equation}
Substituting $\mathbf{w}^{*}$ into \eqref{P1_rate} we have 
$\left\|\mathbf{h}^{H}{\mathbf{G}}\mathbf{A}\mathbf{w}^{*}\right\|^{2}
=P\left\|\mathbf{h}^{H}{\mathbf{G}}\mathbf{A}\right\|^{2}$. 
Thus, (P1) can be rearranged as the following pinching antenna activation optimization problem:
\begin{subequations}\label{P1-A}
    \begin{align*}
        \min_{\mathbf{A}, P, \{L_{n}^{\mathrm{s}}\}}~ 
        & P, \tag{\ref{P1-A}{a}} \label{P1-A_obj}  \\
        \text{s.t.} ~
        & a_{l} \in \{0,1\},  ~ \forall l \in \mathcal{L}, ~ \forall n\in\mathcal{N}, \tag{\ref{P1-A}{b}} \label{P1-A_act} \\   
        & \sum_{l\in\mathcal{L}} a_{l,n} = L^{\mathrm{s}}, ~ \forall n\in\mathcal{N}, \tag{\ref{P1-A}{c}} \label{P1-A_number} \\
        & P\sum_{n\in\mathcal{N}}{\left|\mathbf{h}_{n}^{H}\widetilde{\mathbf{G}}_{n}\frac{1}{\sqrt{L_{n}^{\mathrm{s}}}}\mathbf{a}_{n}\right|^{2}} \geqslant \sigma^{2} \gamma^{\min}, 
        \tag{\ref{P1-A}{d}} \label{P1-A_rate}
    \end{align*}
\end{subequations}
where $\widetilde{\mathbf{G}}_{n}\triangleq\text{diag}\left(\widetilde{\mathbf{g}}_{n}\right)$. The minimum transmit power $P^{*}$ of \eqref{P1-A} is achieved at the lower bound given by constraint \eqref{P1-A_rate}:
\begin{equation}
    P^{*}(\mathbf{A},L_{n}^{\mathrm{s}})=\frac{\sigma^{2}\gamma^{\min}}{\sum_{n\in\mathcal{N}}\frac{1}{L_{n}^{\mathrm{s}}}
    \left|\mathbf{h}_{n}^{H}\widetilde{\mathbf{G}}_{n}\mathbf{a}_{n}\right|^{2}}.
\end{equation} 
Thus, problem \eqref{P1-A} is equivalent to maximizing $1/P^{*}$ by solving the following QCQP for any fixed $L_{n}^{\mathrm{s}}$: 
\begin{subequations}\label{P1-1}
    \begin{align*}
        \text{(P1-1)} ~ \max_{\mathbf{A}}~ 
        &  f(\mathbf{A})
        \!=\!\frac{1}{\sigma^{2}\gamma^{\min}}\!\sum_{n\in\mathcal{N}}\!\frac{1}{L_{n}^{\mathrm{s}}}\!\left|\mathbf{h}_{n}^{H}\widetilde{\mathbf{G}}_{n}\mathbf{a}_{n}\right|^{2}, 
        \tag{\ref{P1-1}{a}} \\
        \text{s.t.} ~
        & \eqref{P1-A_act}, \eqref{P1-A_number}. 
    \end{align*}
\end{subequations}

Note that the optimal $L_{n}^{\mathrm{s}*}$ lies in a finite set $\left\{1,2,\dots,L\right\}$. 
Hence, we can exhaustively search $L_{n}^{\mathrm{s}}\in\left\{1,2,\dots,L\right\}$ and solve (P1-1) to obtain the optimal $L_{n}^{\mathrm{s}*}$, $\mathbf{A}^{*}$, $\mathbf{w}^{*}$, and $P^{*}$ for problem (P1).
However, since the objective function is to maximize a convex quadratic function and the optimization variable is binary, 
the globally optimal solution of nonconvex QCQP (P1-1) cannot be obtained by convex optimization. 
We resort to BnB method to find the global optimum of (P1-1).

\subsubsection{BnB Principles}
BnB solves a nonconvex minimization problem $\min_{\mathbf{x}} f\left(\mathbf{x}\right)$  
by iteratively partitioning the entire solution space $\mathbf{x}\in\mathcal{B}_{\mathrm{ALL}}$ into a set of smaller subregions $\mathcal{S} = \left\{\mathcal{B}_{1}, \mathcal{B}_{2} \dots, \mathcal{B}_{S}\right\}$. 
These smaller subregions are commonly referred to as \textit{boxes}.
Over each box, BnB solves a convex relaxation to evaluate feasibility and compute lower and upper bounds of the original nonconvex problem. 
Boxes that cannot contain the global optimum are pruned to reduce computational overhead. 
As the partitioning proceeds and boxes' sizes shrink, the global upper and lower bounds are progressively tightened and eventually converge to the global optimum \cite{BnB,BnB_branching}.

A box is a $B$-dimension hyperrectangle of variable $\mathbf{x}\in\mathbb{R}^{B}$ with 
lower bound $\underline{\mathbf{b}}=\left[\underline{b}_{1}, \underline{b}_{2}, \dots, \underline{b}_{B}\right]^{T}$ 
and upper bound $\overline{\mathbf{b}}=\left[\overline{b}_{1}, \overline{b}_{2}, \dots, \overline{b}_{B}\right]^{T}$, which is defined as
\begin{equation}
    \mathcal{B} \!\triangleq\! [\underline{\mathbf{b}}, \overline{\mathbf{b}}] \!=\! \left\{ \mathbf{b} \in \mathbb{R}^B \mid \underline{b}_i \leq x_i \leq \overline{b}_i,~ \forall i = 1, \dots, B \right\}.
\end{equation} 
BnB evaluates the lower and upper bounds of the local optimal objective value $f^{*}\left(\mathcal{B}\right)$ within each box $\mathcal{B}$ 
using bounding estimate functions $f_{\mathrm{LB}}\left(\mathcal{B}\right)$ and $f_{\mathrm{UB}}\left(\mathcal{B}\right)$, such that
\begin{itemize}
    \item[i)] $f_{\mathrm{LB}}\left(\mathcal{B}\right) \leqslant  f^{*}\left(\mathcal{B}\right) \leqslant f_{\mathrm{UB}}\left(\mathcal{B}\right)$.
    \item[ii)] $f_{\mathrm{UB}}\left(\mathcal{B}\right) - f_{\mathrm{LB}}\left(\mathcal{B}\right)$ vanishes as $\mathcal{B}$ shrinks to a point.
\end{itemize}
The global upper bound $\mathrm{GUB}$ is the lowest $f_{\mathrm{UB}}\left(\mathcal{B}\right)$ found so far, 
which is progressively reduced as tighter bounds are obtained from smaller feasible subregions. 
Moreover, the global lower bound, defined as $\mathrm{GLB} = \min_{\mathcal{B}'\in\mathcal{S}}\left\{f_{\mathrm{LB}}\left(\mathcal{B}'\right)\right\}$, 
is iteratively refined by pruning redundant boxes that cannot contain optimal solutions and shrinking the remaining boxes. 
As sizes of boxes in $\mathcal{S}$ diminish in $T$ iterations, 
the bound gap $f_{\mathrm{UB}}\left(\mathcal{B}\right)-f_{\mathrm{LB}}\left(\mathcal{B}\right)$ converges to $0$, 
and the globally optimal $f^{*}$ can be approximated as
\begin{equation*}
    \mathrm{GLB}[1] \!\leqslant\! 
    \!\dots\! \leqslant\!\mathrm{GLB}[T]\!\leqslant\!f^{*}\!\leqslant\!\mathrm{GUB}[T]\!\leqslant\!\dots\!\leqslant\!\mathrm{GUB}[1]. 
\end{equation*} 

To enable effective bound estimate, we first construct convex relaxation of problem (P1-1), 
and then develop BnB algorithm to obtain optimal solution.

\subsubsection{Convex Relaxation}
The objective function of problem (P1-1) can be rewritten as 
\begin{equation}
    f\left(\mathbf{A}\right) = \frac{1}{\sigma^{2}\gamma^{\min}}\sum_{n\in\mathcal{N}}
    \frac{1}{L_{n}^{\mathrm{s}}} 
    \mathbf{h}_{n}^{H}\widetilde{\mathbf{G}}_{n}\mathbf{a}_{n}\mathbf{a}_{n}^{T}\widetilde{\mathbf{G}}_{n}\mathbf{h}_{n}.
\end{equation}
We newly introduce variables $\mathbf{\mathbf{Q}}_{n}\in\mathbb{R}^{L\times L}$, which satisfies
\begin{equation}\label{bilinear_SU}
\mathbf{Q}_{n}=\mathbf{a}_{n}\mathbf{a}_{n}^{T}= [a_{l,n}a_{l',n}], \forall n \in \mathcal{N}.
\end{equation} 
The convex hull of the bilinear term $[a_{l}a_{l'}]$ can be obtained by the McCormick envelope, as stated as follows. 
\begin{definition}\label{lemma:ME}
    The convex hull of a bilinear constraint $z = xy$ can be given by the McCormick envelope \cite{MERelaxation} as follows:
    \begin{subequations}
        \begin{equation}\label{MERelaxation_1}
        z\geqslant x\underline{y}+\underline{x}y-\underline{y}\underline{x},
        \end{equation}
        \begin{equation}\label{MERelaxation_2}
        z\geqslant x\overline{y}+\overline{x}y-\overline{y}\overline{x},
        \end{equation}
        \begin{equation}\label{MERelaxation_3}
        z\leqslant x\underline{y}+\overline{x}y-\underline{y}\overline{x},
        \end{equation}
        \begin{equation}\label{MERelaxation_4}
        z\leqslant x\overline{y}+\underline{x}y-\overline{y}\underline{x},
        \end{equation}
    \end{subequations}
    where 
    \begin{equation}
        \underline{x} \leqslant x \leqslant \overline{x}, \quad  
        \underline{y} \leqslant y \leqslant \overline{y}. 
    \end{equation}
\end{definition}

\begin{remark}
    When $\underline{x} = \overline{x}$ and $\underline{y}=\overline{y}$, the equalities in constraints \eqref{MERelaxation_1} -  \eqref{MERelaxation_4} hold true. 
    Hence, the McCormick envelope reduces to the bilinear constraints as the box size reduces, thereby ensuring a tight approximation at convergence.
\end{remark}

Using McCormick envelope, the convex relaxation of \eqref{bilinear_SU} can be given by
\begin{subequations}\label{MESU}
    \begin{equation}\label{MESU_1}
        \mathbf{Q}_{n} \geqslant \mathbf{a}_{n}\underline{\mathbf{a}}_{n}^{T}+\underline{\mathbf{a}}_{n}\mathbf{a}_{n}^{T}-\underline{\mathbf{a}}_{n}
        \underline{\mathbf{a}}_{n}^{T},
    \end{equation}
    \begin{equation}\label{MESU_2}
        \mathbf{Q}_{n} \geqslant \mathbf{a}_{n}\overline{\mathbf{a}}_{n}^{T}+\overline{\mathbf{a}}_{n}
        \mathbf{a}_{n}^{T}-\overline{\mathbf{a}}_{n}\overline{\mathbf{a}}_{n}^{T},
    \end{equation}
    \begin{equation}\label{MESU_3}
        \mathbf{Q}_{n} \leqslant \mathbf{a}_{n}\underline{\mathbf{a}}_{n}^{T}+\overline{\mathbf{a}}_{n}\mathbf{a}_{n}^{T}-\underline{\mathbf{a}}_{n}\overline{\mathbf{a}}_{n}^{T},
    \end{equation}
    \begin{equation}\label{MESU_4}
        \mathbf{Q}_{n} \leqslant \mathbf{a}_{n}\overline{\mathbf{a}}_{n}^{T}+\underline{\mathbf{a}}_{n}\mathbf{a}_{n}^{T}-\overline{\mathbf{a}}_{n}\underline{\mathbf{a}}_{n}^{T},
    \end{equation}
\end{subequations}
\begin{equation}\label{MESU_bound}
    \underline{\mathbf{a}}_{n} \leqslant \mathbf{a}_{n} \leqslant \overline{\mathbf{a}}_{n}.
\end{equation}
Hence, the nonconvex QCQP (P1-1) can be relaxed into the following convex optimization problem:
\begin{subequations}\label{P1-C}
    \begin{align*}
        \text{(P1-C)} ~ \max_{\mathbf{A},\mathbf{Q}}~ 
        &  \sum_{n\in\mathcal{N}}
        \frac{1}{L_{n}^{\mathrm{s}}\sigma^{2}\gamma^{\min}} 
        \mathbf{h}_{n}^{H}\widetilde{\mathbf{G}}_{n}\mathbf{Q}\widetilde{\mathbf{G}}_{n}\mathbf{h}_{n}, \\
        \text{s.t.} ~
        & \eqref{P1-A_number}, \eqref{MESU_1}-\eqref{MESU_4},\eqref{MESU_bound}.
    \end{align*}
\end{subequations}

Define $\mathcal{X}$ and $\mathcal{C}$ as the feasible sets for $\mathbf{A}$ in original problem (P1-1) and relaxed problem (P1-C), respectively. 
Moreover, $\mathcal{B} = \left\{ \mathbf{A} \mid \underline{\mathbf{a}}_{n} \leqslant \mathbf{a}_{n} \leqslant \overline{\mathbf{a}}_{n}, 
\forall n\in\mathcal{N} \right\}$ is a box for $\mathbf{A}$ corresponding to constraint \eqref{MESU_bound}. 
The feasibility and bounds of problem (P1-1) can be effectively estimated as follows. 
\begin{proposition}\label{proposition:bounding_property}
    As $\mathcal{C}$ is a convex hull of $\mathcal{X}$, i.e., $\mathcal{X} \subseteq  \mathcal{C}$, the feasibility and bounds of (P1-1) satisfy:
    \begin{enumerate}
    \item[i)] \textit{(Feasibility)} If the relaxed problem (P1-C) is infeasible over $\mathcal{B} \cap \mathcal{C}$, then (P1-1) is infeasible over $\mathcal{B} \cap \mathcal{X}$.
    \item[ii)] \textit{(Lower bound)} The optimum value $f_{\mathrm{c}}$ of relaxed problem (P1-C) is a lower bound of optimum $f^{*}$ of 
        the original problem (P1-1) , i.e., $f_{\mathrm{c}} \leqslant f^{*}$. 
        The equality holds when $\mathcal{B}$ shrinks to a discrete point.
    \item[iii)] \textit{(Upper bound)} Any feasible solution $\mathbf{x}\in\mathcal{X}$ gives an upper bound $f\left(\mathbf{x}\right) \geqslant f^{*}$ of original problem (P1-1). 
    \end{enumerate}

\end{proposition}

\begin{algorithm}[!tp]
    \caption{Optimal Beamforming for Single-User Scenario}
    \label{alg:branch_and_bound_SU}
    \begin{algorithmic}[1]
    \REQUIRE Channel $\mathbf{h}$, $\mathbf{G}$, tolerance threshold $\epsilon>0$.
    \STATE Initialize box $\mathcal{B}=\big[\underline{\mathbf{b}}, \overline{\mathbf{b}}\big]= 
    \big[\mathbf{0}_{M\times 1},\mathbf{1}_{M \times 1}\big]$, and box list $\mathcal{S} =\left\{\mathcal{B}\right\}$.
    \STATE Set $\mathrm{GUB} = +\infty$, $\mathrm{GLB} = - \infty$, and $P^{*}=+\infty$.
    \FOR{$L_{n}^{\mathrm{s}}\in \{1,2,\dots, L\}, \forall n$}
        \WHILE{$\mathcal{S}\ne\emptyset$ and $\mathrm{GUB}-\mathrm{GLB} > \epsilon$}
            \item[] {//*~\textit{\textbf{Branching}}:}
            \STATE Select branching box and edge by \eqref{box_selection} and \eqref{edge_selection}. 
            \STATE Obtain $\mathcal{B}_{-}$ and $\mathcal{B}_{+}$ by \eqref{child_boxes}. Update $\mathcal{S}$ by \eqref{box_list}.
            \FOR{$\mathcal{B} \in \left\{\mathcal{B}_{-}, \mathcal{B}_{+}\right\}$}
                \item[] {//*~\textit{\textbf{Bounding}}:}
                \STATE If (P1-C) is infeasible, prune $\mathcal{B}$, turn to Line 7.
                \STATE Update $f_{\mathrm{LB}}\left(\mathcal{B}\right)$ by \eqref{child_boxes_SU}.
                \STATE Obtain $\mathbf{A}_{\mathrm{prj}}$. Compute $f_{\mathrm{UB}}$,  $\mathbf{w}_{\mathrm{prj}}$ by \eqref{LUB_SU}, \eqref{BF_SU}. 
                \STATE If $f_{\mathrm{UB}}\!<\!\mathrm{GUB}$, update $\mathbf{x}^{*} \!=\!  \left\{\mathbf{A}_{\mathrm{prj}}, 
                \mathbf{w}_{\mathrm{prj}}\right\}$.
                \STATE Update $\mathrm{GLB}$ and $\mathrm{GUB}$ by \eqref{GLB_SU} and \eqref{GUB_SU}.
                \item[] {//*~\textit{\textbf{Pruning}}:}
                \STATE Prune $\mathcal{B}$ if it meets fathomed condition \eqref{fathomed_box}. 
                \STATE Prune non-optimal boxes $\mathcal{B}'\in\mathcal{S}$ satisfying \eqref{nonoptimal_box}.
            \ENDFOR
        \ENDWHILE
        \STATE If $\mathrm{GUB} < P^{*}$, update $P^{*}$ and the optimal $L_{n}^{\mathrm{s}*}, \mathbf{A}^{*}, \mathbf{w}^{*}$.
    \ENDFOR
    \ENSURE Optimal $L_{n}^{\mathrm{s}*}, \mathbf{A}^{*}, \mathbf{w}^{*}$, and $P^{*}$.
    \end{algorithmic}
\end{algorithm}

\subsubsection{BnB Algorithm}
To optimally solve problem (P1-1), we perform branching over discrete variables $\mathbf{A}=[a_{l,n}]$.  
Hence, the $B$-dimension branching variables are given by $\mathbf{b} = \left[\mathbf{a}_{1}^{T}, \mathbf{a}_{2}^{T}, \dots, 
\mathbf{a}_{N}^{T}\right]^{T}\in\mathbb{R}^{B\times 1}$ with $B=M$.  
From definitions, the initial space region $\mathcal{B}_{\mathrm{ALL}}$ can be given by $\underline{\mathbf{b}}=0$ and $\overline{\mathbf{b}}=1$.
The BnB procedure performs the following branching, bounding, and pruning steps in each iteration:

\textit{(i) Branching}: 
At each iteration, we select a box $\mathcal{B}_{o}$ from the candidate list $\mathcal{S}$,  
and branch it into two children boxes $\mathcal{B}_{-}$ and $\mathcal{B}_{+}$ along a certain edge $e\in\left\{1,2,\dots,B\right\}$. 
We exploit best-bound-first (BBF) box selection rule \cite{BnB_branching}, 
where the box achieving the best lower bound is chosen to branch:
\begin{equation}\label{box_selection}
    \mathcal{B}_{o} = \mathop{\arg\min}_{\mathcal{B}\in\mathcal{S}} ~ f_{\mathrm{LB}}\left(\mathcal{B}\right).
\end{equation}
Moreover, the maximum-length-first (MLF) edge selection rule is adopted. 
The selected box $\mathcal{B}_{o}$ is equally divided into $\mathcal{B}_{-}$ and $\mathcal{B}_{+}$ 
along its longest edge 
\begin{equation}\label{edge_selection}
    e = \mathop{\arg\max}\limits_{i \in \left\{1,2,\dots,B\right\}} \phi_{i} = \mathop{\arg\max}\limits_{i \in \left\{1,2,\dots,B\right\}}  |\overline{b}_{i} - \underline{b}_{i}|,
\end{equation}
where $\phi_{i}$ is the length of the $i$-th edge of box $\mathcal{B}_{o}$. 
The resultant boxes can be defined as $\mathcal{B}_{-} = \left[\underline{\mathbf{b}}, \overline{\mathbf{b}}_{\mathrm{new}}\right]$ 
and $\mathcal{B}_{+} = \left[\underline{\mathbf{b}}_{\mathrm{new}}, \overline{\mathbf{b}}\right]$, 
where the new corner points $\overline{\mathbf{b}}_{\mathrm{new}}$ and $\underline{\mathbf{b}}_{\mathrm{new}}$ are given by
\begin{subequations}\label{child_boxes_SU}
    \begin{equation}
        \overline{\mathbf{b}}_{\mathrm{new}}\!=\!
            \left[\overline{b}_{1},\overline{b}_{2},\dots,\overline{b}_{e-1},0,\overline{b}_{e+1},\dots,\overline{b}_{B}\right]^{T}\!,
    \end{equation} 
    \begin{equation}
        \underline{\mathbf{b}}_{\mathrm{new}}\!=\!
            \left[\underline{b}_{1},\underline{b}_{2},\dots,\underline{b}_{e-1},1,\underline{b}_{e+1},\dots,\underline{b}_{B}\right]^{T}\!.
    \end{equation} 
\end{subequations}
Then, the candidate box list $\mathcal{S}$ is updated by 
\begin{equation}\label{box_list}
    \mathcal{S} \leftarrow \mathcal{S} \setminus \left\{\mathcal{B}_{o}\right\} \cup \left\{\mathcal{B}_{-},\mathcal{B}_{+}\right\}.
\end{equation}

\textit{(ii) Bounding:} 
We evaluate the bounds of $P^{*}$ over each children box $\mathcal{B} \in \left\{\mathcal{B}_{+}, \mathcal{B}_{-}\right\}$. 
Let $\mathbf{x}_{\mathrm{c}} = \left\{\mathbf{A}_{\mathrm{c}}, \mathbf{Q}_{\mathrm{c}}\right\}$
denote the optimal solution of the relaxed problem (P-C) within box $\mathcal{B}$. 
According to property ii) in \textbf{Proposition \ref{proposition:bounding_property}}, 
the optimal value of relaxed problem (P1-C) provides a valid lower bound 
for the original problem (P1), i.e.,
\begin{equation}\label{LLB_SU}
    f_{\mathrm{LB}}\left(\mathcal{B}\right) \!= P^{*} \left(  \mathbf{A}_{\mathrm{c}},L_{n}^{\mathrm{s}}\right)\!=
    \!\frac{\sigma^{2}\gamma^{\min}}{\sum\limits_{n\in\mathcal{N}}\frac{1}{L_{n}^{\mathrm{s}}}\!\left|\mathbf{h}_{n}^{H}\widetilde{\mathbf{G}}_{n}\mathbf{a}_{n,\mathrm{c}}^{*}\right|^{2}}.
\end{equation}
The global lower bound $\mathrm{GLB}$ is then updated by the minimum $f_{\mathrm{LB}}\left(\mathcal{B}'\right)$ over all candidate boxes $\mathcal{B}'\in\mathcal{S}$:
\begin{equation}\label{GLB_SU}
    \mathrm{GLB} = \min_{\mathcal{B}' \in \mathcal{S}} f_{\mathrm{LB}}\left(\mathcal{B}'\right).
\end{equation}

By projecting the relaxed pinching antenna activation solution $\mathbf{A}_{\mathrm{c}}$ into binary variables  $\mathbf{A}_{\mathrm{prj}}$ through rounding operations, 
we can further obtain a feasible solution of the original problem and evaluate the upper bounds of $\mathcal{B}$ as
\begin{equation}\label{LUB_SU}
    f_{\mathrm{UB}}\!\left(\mathcal{B}\right) \!=\! P^{*}(\mathbf{A}_{\mathrm{prj}},L_{n}^{\mathrm{s}})\!=
    \!\frac{\sigma^{2}\gamma^{\min}}{\sum\limits_{n\in\mathcal{N}}\frac{1}{L_{n}^{\mathrm{s}}}\!\left|\mathbf{h}_{n}^{H}\widetilde{\mathbf{G}}_{n}\mathbf{a}_{n,\mathrm{prj}}\right|^{2}}.
\end{equation}
The global upper bound $\mathrm{GUB}$ can be updated by the best feasible solution currently found:
\begin{equation}\label{GUB_SU}
    \mathrm{GUB} \leftarrow \min\!\left\{\mathrm{GUB}, f_{\mathrm{UB}}\left(\mathcal{B}_{-}\right),f_{\mathrm{UB}}\left(\mathcal{B}_{+}\right)\right\}.
\end{equation}

\textit{(iii) Pruning:}
Boxes that could not contain optimal solutions  can be identified and pruned, 
thus accelerating convergence without impacting the global optimality. 
Specifically, the following boxes can be pruned from $\mathcal{S}$:
\begin{itemize}
    \item {Infeasible}:  If the relaxed problem (P1-C) is infeasible over $\mathcal{B}$, 
    $\mathcal{B}$ is infeasible for problem (P1-1) according to property i) in \textbf{Proposition \ref{proposition:bounding_property}}, 
    and can be pruned from $\mathcal{S}$. 
    \item {Fathomed}: $\mathcal{B}$ can be pruned if it is fathomed (fully explored), i.e., local upper and lower bounds satisfy
    \begin{equation}\label{fathomed_box}
        f_{\mathrm{UB}}\left(\mathcal{B}\right) - f_{\mathrm{LB}}\left(\mathcal{B}\right) \leqslant \varepsilon, ~\forall \mathcal{B} \in \left\{\mathcal{B}_{-},\mathcal{B}_{+}\right\}.
    \end{equation} 
    \item {Nonoptimal}: For any box $\mathcal{B}'$ in $\mathcal{S}$, if the local lower bound $f_{\mathrm{LB}}\left(\mathcal{B}'\right)$ 
    exceeds the global upper bound $\mathrm{GUB}$, i.e., 
    \begin{equation}\label{nonoptimal_box}
        f_{\mathrm{LB}}\left(\mathcal{B}'\right) > \mathrm{GUB}, ~ \forall \mathcal{B}'\in\mathcal{S},
    \end{equation}
    $\mathcal{B}$ must not contain optimal solutions and can be pruned.
\end{itemize}

Algorithm \ref{alg:branch_and_bound_SU} summarizes the entire procedure to obtain the globally optimal design for PASS. 
The best-case complexity and the worst-case complexity of BnB search are given by $\mathcal{O}\left(M\right)$ and $\mathcal{O}\left(2^{M}\right)$, respectively.

\subsection{Optimal Solution for Multi-User Scenario}
This part investigates the optimal solution of problem (P0) for the general multi-user scenario. 
Unlike the single-user case, the optimal multi-user transmit beamforming cannot be expressed by explicit closed-form expressions to simplify the optimization problem. 
Alternatively, we construct a tractable convex relaxation problem to handle sophisticated coupling among $\mathbf{W}$, $\mathbf{A}$, and $L_{n}^{\mathrm{s}}$, 
and develop a tailored BnB algorithm. 
\subsubsection{Convex Relaxation}
By exploiting the phase-rotation invariance property of modulus operations  \cite{BF_Emil}, the minimum SINR constraints in \eqref{P1_rate} 
can be equivalently reformulated as
\begin{multline}\label{P2_SINR}
    \hspace{-1em}\!\mathrm{Re}\!\left\{ \mathbf{h}_{k}^{H}\mathbf{G}\mathbf{A}\mathbf{w}_{k}\right\} \! \geqslant \! \!
\sqrt{\gamma^{\min}\!\!\left(\sum\limits _{k'\ne k}\!\left|\mathbf{h}_{k}^{H}\!\mathbf{G}\mathbf{A}\mathbf{w}_{k'}\right|^{2}\!+\!\sigma^{2}\!\right)}, \forall k \in \mathcal{K},
\\\mathrm{Im}\left\{ \mathbf{h}_{k}^{H}\mathbf{G}\mathbf{A}\mathbf{w}_{k}\right\}=0, \forall k \in \mathcal{K}. 
\end{multline}

To make the coupling term \eqref{coupled_term_MU} tractable, we newly introduce an auxiliary transmit beamforming matrix $\mathbf{D} = [d_{n,k}]\in\mathbb{C}^{N\times K}$  
normalized by the equal power radiation ratios. 
Specifically, the $n$-th row vector $\mathbf{D}_{n,:}$ indicates effective transmit beamforming for signals radiated by each pinching antenna over waveguide $n$, 
which is defined as
\begin{equation}\label{w_d_transfer}
    \mathbf{D}_{n,:} = \frac{1}{\sqrt{L_{n}^{\mathrm{s}}}} \mathbf{W}_{n,:}, ~\forall n \in \mathcal{N}.
\end{equation}
Moreover, we define a new set of slack variables $\mathbf{z}_{k}\in\mathbb{C}^{M\times 1}$ to represent the bilinear terms $\mathbf{A}\mathbf{d}_{k}$: 
\begin{equation}\label{bilinear_MU}
    \mathbf{z}_{k} = \mathbf{A} \mathbf{d}_{k} = \left[\mathbf{a}_{1}d_{1,k}, \mathbf{a}_{2}d_{2,k}, \dots, \mathbf{a}_{N}d_{N,k} \right],  ~ \forall k \in \mathcal{K}.
\end{equation}
Substituting \eqref{w_d_transfer} and \eqref{bilinear_MU} into \eqref{coupled_term_MU}, we can obtain 
$\mathbf{G} \mathbf{A} \mathbf{w}_{k} =   \mathbf{\widetilde{G}} \mathbf{A} \mathbf{d}_{k}=   \mathbf{\widetilde{G}} \mathbf{z}_{k}$.
Based on the above definitions, constraint \eqref{P2_SINR} can be equivalently converted into the following convex second-order cone  (SOC) constraints:
\begin{equation}\label{P2_convex_SINR}\begin{split}
    &\mathrm{Re}\! \left\{ \mathbf{h}_{k}^{H}\! \widetilde{\mathbf{G}}\mathbf{z}_{k}\right\} \! \geqslant \! 
\sqrt{\gamma^{\min}\!\!\left(\sum\limits _{k'\ne k}\!\left|\mathbf{h}_{k}^{H}\!\widetilde{\mathbf{G}}\mathbf{z}_{k'}\right|^{2}\!\!+\!\!\sigma^{2}\right)}, 
~ \forall k \in \mathcal{K},
    \\& \mathrm{Im}\left\{ \mathbf{h}_{k}^{H}\widetilde{\mathbf{G}}\mathbf{z}_{k}\right\} =0, ~ \forall k \in \mathcal{K}.
\end{split}
\end{equation}

Furthermore, the objective function can be rearranged as
\begin{equation}\label{MU_obj_convert}
    \left\Vert \mathbf{W}\right\Vert _{F}^{2}
    \overset{(a)}{=}\sum_{k\in\mathcal{K}}\sum_{n\in\mathcal{N}}\left\|\mathbf{a}_{n}d_{n,k}\right\|^{2}
    \!\overset{\eqref{bilinear_MU}}{=}\!\left\|\mathbf{Z}\right\|_{F}^{2},
\end{equation}
where matrix $\mathbf{Z}=\left[\mathbf{z}_{1},\mathbf{z}_{2},\dots,\mathbf{z}_{K}\right]\in\mathbb{C}^{M\times K}$ stacks vectors $\mathbf{z}_{k}^{T}$, $\forall k \in\mathcal{K}$.  
In \eqref{MU_obj_convert}, the equality (a) results from the fact that the binary variable satisfies $a_{l,n} = a_{l,n}^{2} \geqslant 0$, and thus 
\begin{equation} 
    \left|w_{n,k}\right|^{2}
    \!\overset{\eqref{w_d_transfer}}{=}\! L_{n}^{\mathrm{s}} \left|d_{n,k}\right|^{2}
    \!=\!\sum_{l=1}^{L}\!a_{l,n}^{2}\left|d_{n,k}\right|^{2}
    \!=\!\sum_{l=1}^{L}\!\left|a_{l,n}d_{n,k}\right|^{2}\!.
\end{equation}
Hence, problem (P0) can be equivalently reformulated as 
    \begin{equation*}\label{P2}
        \text{(P2)} ~   \min_{\mathbf{A}, \mathbf{D},  \mathbf{Z}}
        F\left(\mathbf{A}, \mathbf{D},  \mathbf{Z}\right) = \left\Vert \mathbf{Z}\right\Vert_{F} ^{2}, 
        ~ \text{s.t.} ~
        \eqref{P_act}, \eqref{bilinear_MU}, \eqref{P2_convex_SINR}.
    \end{equation*}

We derive the McCormick envelope for the complex-value bilinear equality constraint \eqref{bilinear_MU} as
\begin{subequations}\label{McCormickE_mu}
    \begin{equation}
        \mathrm{Re}\!\left\{ \mathbf{Z}\right\} \! \geqslant \! \mathbf{A}\underline{\mathbf{U}} \!+\!\underline{\mathbf{A}}\mathbf{U} \!-\! \underline{\mathbf{A}}\underline{\mathbf{U}},
        ~
        \mathrm{Im}\!\left\{ \mathbf{Z}\right\} \! \geqslant \! \mathbf{A}\underline{\mathbf{V}} \!+\! \underline{\mathbf{A}}\mathbf{V} \!-\! \underline{\mathbf{A}}\underline{\mathbf{V}},
    \end{equation}
    \begin{equation}
        \mathrm{Re}\!\left\{ \mathbf{Z}\right\} \! \geqslant \! \mathbf{A}\overline{\mathbf{U}} \!+\!\overline{\mathbf{A}}\mathbf{U} \!-\! \overline{\mathbf{A}}\overline{\mathbf{U}},
        ~
        \mathrm{Im}\!\left\{ \mathbf{Z}\right\} \! \geqslant \! \mathbf{A}\overline{\mathbf{V}} \!+\! \overline{\mathbf{A}}\mathbf{V} \!-\! \overline{\mathbf{A}}\overline{\mathbf{V}},
    \end{equation}
    \begin{equation}
        \mathrm{Re}\!\left\{ \mathbf{Z}\right\} \! \geqslant \! \mathbf{A}\underline{\mathbf{U}} \!+\!\overline{\mathbf{A}}\mathbf{U} \!-\! \overline{\mathbf{A}}\underline{\mathbf{U}},
        ~
        \mathrm{Im}\!\left\{ \mathbf{Z}\right\} \! \geqslant \! \mathbf{A}\underline{\mathbf{V}} \!+\! \overline{\mathbf{A}}\mathbf{V} \!-\! \overline{\mathbf{A}}\underline{\mathbf{V}},
    \end{equation}
    \begin{equation}
        \mathrm{Re}\!\left\{ \mathbf{Z}\right\} \! \geqslant \! \mathbf{A}\overline{\mathbf{U}} \!+\!\underline{\mathbf{A}}\mathbf{U} \!-\! \underline{\mathbf{A}}\overline{\mathbf{U}},
        ~
        \mathrm{Im}\!\left\{ \mathbf{Z}\right\} \! \geqslant \! \mathbf{A}\overline{\mathbf{V}} \!+\! \underline{\mathbf{A}}\mathbf{V} \!-\! \underline{\mathbf{A}}\overline{\mathbf{V}}, 
    \end{equation}
    \begin{equation}\label{varbound_MU}
        \underline{ \mathbf{U} } \leqslant \mathbf{U} \leqslant \overline{ \mathbf{U} }, \quad  
        \underline{ \mathbf{V} } \leqslant \mathbf{V} \leqslant \overline{ \mathbf{V} }, \quad 
        \underline{\mathbf{a}}_{n} \leqslant \mathbf{a}_{n} \leqslant \overline{\mathbf{a}}_{n},
    \end{equation}
\end{subequations}
where $\mathbf{A} \triangleq \text{blkdiag} \left(\mathbf{a}_{1}, \mathbf{a}_{2}, \dots, \mathbf{a}_{N}\right) \in \mathbb{Z}^{M \times N}$, 
and auxiliary variables $\mathbf{U}\in\mathbb{R}^{N\times K}$ and $\mathbf{V}\in\mathbb{R}^{N\times K}$ are defined as 
$\mathbf{U} \!=\! \mathrm{Re}\left\{\mathbf{D}\right\}$ and 
$\mathbf{V} \!=\! \mathrm{Im}\left\{\mathbf{D}\right\}$, respectively.
Therefore, the convex relaxation problem of (P2) can be constructed by relaxing binary constraints and replacing bilinear constraints with the McCormick envelope: 
\begin{equation*}\label{P2_convex}
    \text{(P2-C)} ~ \min_{\mathbf{A}, \mathbf{D},  \mathbf{Z}}
    ~ F\left(\mathbf{A}, \mathbf{D},  \mathbf{Z}\right) =  \left\Vert \mathbf{Z}\right\Vert_{F} ^{2}, ~
    \text{s.t.} ~ 
    \eqref{P2_convex_SINR}, \eqref{McCormickE_mu}.
\end{equation*}

\subsubsection{BnB Algorithm}
We perform branching over the $B$-dimension variables $\mathbf{b} = \big[\mathbf{a}_{1}^{T}, \mathbf{a}_{2}^{T}, \dots, 
\mathbf{a}_{N}^{T}, \mathrm{vec}\big(\mathrm{Re}\left\{\mathbf{D}\right\}\big), \allowbreak \mathrm{vec}\big(\mathrm{Im}\left\{\mathbf{D}\right\}\big)\big]^{T}\in\mathbb{R}^{B\times 1}$ 
to determine variable bounds in \eqref{varbound_MU}, where $B = M+2NK$.  
From definitions, the initial feasible region of $\mathbf{A}$ can be determined as $\underline{\mathbf{a}}_{n}=0$ and $\overline{\mathbf{a}}_{n}=1$, $\forall n$.
Moreover, the transmit beamforming coefficients are bounded by  $\underline{\mathbf{u}}_{k} = \underline{\mathbf{v}}_{k}= - \sqrt{P_0}$ 
and $\overline{\mathbf{u}}_{k} = \overline{\mathbf{v}}_{k}= \sqrt{P_0}$, 
where $P_{0}$ denotes a sufficiently large power budget for each pinching antenna such that problem (P2-C) is feasible. 

\begin{algorithm}[!tp]
    \caption{Optimal Beamforming for Multi-User Scenario}
    \label{alg:branch_and_bound_MU}
    \begin{algorithmic}[1]
    \REQUIRE Channel $\mathbf{H}$, $\mathbf{G}$, tolerance threshold $\varepsilon>0$.
    \STATE Initialize $\underline{\mathbf{b}} \!=\! \big[\mathbf{0}_{1\times M}, -\sqrt{P_0}\mathbf{1}_{1\times 2NK}\big]^{T}, 
    \overline{\mathbf{b}} \!=\! \big[\mathbf{0}_{1\times M}, \allowbreak \sqrt{P_0}\mathbf{1}_{1\times 2NK}\big]^{T}\!$,  
    $\mathcal{B}\!=\!\big[\underline{\mathbf{b}}, \overline{\mathbf{b}}\big]$, and $\mathcal{S} \!=\!\left\{\mathcal{B}\right\}$.
    \STATE Initialize $\mathrm{GUB} = +\infty$, $\mathrm{GLB} = - \infty$.
    \WHILE{$\mathcal{S}\ne\emptyset$ and $\mathrm{GUB}-\mathrm{GLB} > \varepsilon$}
        \item[] {//*~\textit{\textbf{Branching}}:}
        \STATE Select branching box $\mathcal{B}_{o}$ and edge $e$ by \eqref{box_selection} and \eqref{edge_selection}. 
        \STATE Obtain $\mathcal{B}_{-}$ and $\mathcal{B}_{+}$ by \eqref{child_boxes}. Update $\mathcal{S}$ by \eqref{box_list}. 
        \FOR{each children box $\mathcal{B} \in \left\{\mathcal{B}_{-}, \mathcal{B}_{+}\right\}$}
            \item[] {//*~\textit{\textbf{Bounding}}:}
            \STATE If (P2-C) is infeasible, prune $\mathcal{B}$ and turn to line 6.
            \STATE Update $f_{\mathrm{LB}}\left(\mathcal{B}\right)$ by \eqref{LLB_MU}. 
            \STATE Compute $\mathbf{A}_{\mathrm{prj}}, \mathbf{D}_{\mathrm{prj}}$ and $f_{\mathrm{UB}}\!\left(\mathcal{B}\right)$. 
            \STATE Update $\mathrm{GLB}$ and $\mathrm{GUB}$ by \eqref{GLB_SU} and \eqref{GUB_SU}.
            \item[] {//*~\textit{\textbf{Pruning}}:}
            \STATE Prune $\mathcal{B}$ if it meets fathomed condition \eqref{fathomed_box}. 
            \STATE Prune non-optimal boxes $\mathcal{B}'\in\mathcal{S}$ satisfying \eqref{nonoptimal_box}.
        \ENDFOR
    \ENDWHILE
    \STATE Obtain the optimal $\mathbf{W}_{n,:}^{*} = \sqrt{L_{n}^{\mathrm{s}}}\mathbf{D}_{n,:}^{*}$, $\forall n \in \mathcal{N}$. 
    \ENSURE Optimal $\mathbf{A}^{*}$, $\mathbf{W}^{*}$, and $f^{*}=\mathrm{GUB}$.
    \end{algorithmic}
\end{algorithm}

The BnB procedure performs the following branching, bounding, and pruning steps in each iteration.
First, a branching box $\mathcal{B}_{o}$ is selected from $\mathcal{S}$ based on BBF box selection rule \eqref{box_selection}.
The selected box is divided into children boxes $\mathcal{B}_{-} = \left[\underline{\mathbf{b}}, \overline{\mathbf{b}}_{\mathrm{new}}\right]$ and 
$\mathcal{B}_{+} = \left[\underline{\mathbf{b}}_{\mathrm{new}}, \overline{\mathbf{b}}\right]$  
along the longest edge $e$ based on MLF rule \eqref{edge_selection}. 
If edge $e$ corresponds to a binary variable in vector $\mathbf{b}$ (i.e., $1\leqslant e \leqslant M$), 
it will be branched into two discrete points at $0$ and $1$. 
Otherwise, edge $e$  represents a continuous variable  (i.e., $e > M$) and is equally divided into two parts.
Hence, $\overline{\mathbf{b}}_{\mathrm{new}}$ and $\underline{\mathbf{b}}_{\mathrm{new}}$ are given by
\begin{equation}\label{child_boxes}
    \begin{split}
        &\overline{\mathbf{b}}_{\mathrm{new}}\!\!=\!\!\begin{cases}
        \!\left[\overline{\mathbf{b}}_{1:\,e-1},~0,~\overline{\mathbf{b}}_{e+1\,:B},\right]^{T}\!, & \text{if $1\!\leqslant\! e\!\leqslant\! M$,}\\
        \!\left[\overline{\mathbf{b}}_{1:\,e-1},~\frac{\overline{b}_{e}+\underline{b}_{e}}{2},~\overline{\mathbf{b}}_{e+1\,:B}\right]^{T}\!, & \text{if $e>M$}.
        \end{cases}
\\&
    \underline{\mathbf{b}}_{\mathrm{new}}\!\!=\!\!\begin{cases}
        \!\left[\underline{\mathbf{b}}_{1:\,e-1},~1,~\underline{\mathbf{b}}_{e+1\,:B}\right]^{T}\!, & \text{if $1\!\leqslant\! e\!\leqslant\! M$},\\
        \!\left[\underline{\mathbf{b}}_{1:\,e-1},~\frac{\overline{b}_{e}+\underline{b}_{e}}{2},~\underline{\mathbf{b}}_{e+1\,:B}\right]^{T}\!, & \text{if $e>M$.}
        \end{cases}
    \end{split}\end{equation}

By solving the convex relaxation problem (P2-C) over the children box $\mathcal{B}\in\left\{\mathcal{B}_{-},\mathcal{B}_{+}\right\}$, 
the local lower bound $f_{\mathrm{LB}}\left(\mathcal{B}\right)$ is evaluated as 
\begin{equation}\label{LLB_MU}
    f_{\mathrm{LB}}\left(\mathcal{B}\right) = F\left(\mathbf{A}_{\mathrm{c}}, \mathbf{D}_{\mathrm{c}},  \mathbf{Z}_{\mathrm{c}}\right) \leqslant f^{*}\left(\mathcal{B}\right),
\end{equation}
where $F\left(\mathbf{A}_{\mathrm{c}}, \mathbf{D}_{\mathrm{c}},  \mathbf{Z}_{\mathrm{c}}\right)$ 
and $f^{*}\left(\mathcal{B}\right)$ denotes the optimal objective values of 
relaxed problem (P2-C) and original problem (P2) over feasible region $\mathcal{B}$, respectively. 
By projecting $\mathbf{A}_{\mathrm{c}}$ into binary variables $\mathbf{A}_{\mathrm{proj}}$, 
and solving the original problem (P2) for the given $\mathbf{A}_{\mathrm{proj}}$, 
a feasible solution $\left\{\mathbf{A}_{\mathrm{proj}},\mathbf{D}_{\mathrm{proj}}\right\}$ can be further 
obtained, which provides a local upper bound $f_{\mathrm{UB}}\left(\mathcal{B}\right)$ of the optimal objective value.
Similar to the single-user case, the global upper and lower bounds $\mathrm{GUB}$ and $\mathrm{GLB}$ can be further refined. 
Then, boxes unnecessary to explore can be identified and pruned from $\mathcal{S}$. 
\textbf{Algorithm \ref{alg:branch_and_bound_MU}} summarizes the entire BnB procedure for multi-user scenario.

\subsubsection{Convergence and Optimality Analysis}
The convergence, optimality, and the worst-case complexity of the proposed BnB in \textbf{Algorithm \ref{alg:branch_and_bound_MU}} can be mathematically proven as follows. 
\begin{lemma}\label{lemma:GLB_GUB_diff}
    The gap $\mathrm{GUB}-\mathrm{GLB}$ vanishes as the maximum edge length $\phi_{\max} \!=\! \max\limits_{i\in\{1,2,..,B\}} \left\{\overline{b}_{i} - \underline{b}_{i}\right\}$ decreases. 
    Given any tolerance $\varepsilon >0$,  if $\phi_{\max}$ becomes smaller than the threshold 
    \begin{equation}\label{eq:GLB_GUB_diff}
        \phi_{\max}  \leqslant \xi \triangleq \frac{\varepsilon}{\sqrt{2MP_{0}B}},
    \end{equation}
    then $\mathrm{GUB}-\mathrm{GLB} \leqslant \varepsilon$ and Algorithm \ref{alg:branch_and_bound_MU} terminates. 
    \begin{proof}
        See Appendix \ref{proof:lemma:GLB_GUB_diff}.
    \end{proof}
\end{lemma}

\begin{theorem}\label{theorem:MU_BnBConverge}
The proposed BnB algorithm converges to an $\varepsilon$-optimal solution in finite iterations. 
That is, the achieved optimum value $f^{*}$ can be arbitrarily close to the true optimum $f_{\mathrm{true}}^{*}$ of problem (P2), 
i.e., $f^{*} \leqslant f_{\mathrm{true}}^{*} + \varepsilon$, $\forall \varepsilon \geqslant 0$. 
\begin{proof}
    See Appendix \ref{proof:theorem:MU_BnBConverge}.
\end{proof}
\end{theorem}

\begin{theorem}\label{theorem:BnB_complexity}
    An $\varepsilon$-optimal solution can be obtained in at most
    $T_{\max} = \left\lceil \frac{\psi_{\mathrm{vol}}}{\xi^{2NK}}2^{B+1} - 1 \right\rceil$ branching 
    iterations, where $\psi_{\mathrm{vol}} \!=\! \big( 2\sqrt{P_{0}} \big) ^{2NK}$ denotes the volume of the initial feasible region for continuous variables
    at the root node.
    \begin{proof}
        See Appendix \ref{proof:theorem:BnB_complexity}.
    \end{proof}
\end{theorem}
The complexity of solving problem (P2-C) using interior point method is $\mathcal{O}\left((M(K+1)+2NK)^{3.5}\right)$ \cite{InteriorPoint_Iteration,InteriorPointComplexity}. 
Hence, the worst-case complexity of \textbf{Algorithm \ref{alg:branch_and_bound_MU}} is given by $\mathcal{O}\left(2T_{\max}(M(K+1)+2NK)^{3.5}\right)$.

\section{Low-complexity Many-to-many Matching-Based Suboptimal Solution}
In this section, we develop a low-complexity welfare-driven many-to-many matching algorithm to solve problem (P0). 

\subsection{Many-to-Many Waveguide-Pinch Matching Model}
We optimize discrete activation $\mathbf{A}$ of pinching antennas in problem (P0) by the matching theory, and determine transmit beamforming $\mathbf{W}$ for the given $\mathbf{A}$ by the KKT theory. 

The pinching antenna activation can be modelled as a many-to-many matching game $\mu$ between two sets of agents: 
the waveguides $\mathcal{N}=\left\{1,2,\dots,N\right\}$, and 
the pinching antenna indices  $\mathcal{L}=\left\{1,2,\dots, L\right\}$. 
Each pinching antenna index $l \in \mathcal{L}$ denotes the same order of pinching antenna along different waveguides, 
but corresponds to physically distinct pinching antennas. 
Thus, $\mathcal{L}$ serves as a logical index set for modelling matching relationships.
A waveguide-pinching antenna pair $(l,n)$ represents the $l$-th pinching antenna at waveguide $n$, and a match between them indicates the activation of that pinching antenna. 
A waveguide-pinching antenna matching $\mu$ is defined as a mapping from waveguides to subsets of pinching antenna indices, 
which satisfies the following conditions: 
(i) Each pinching antenna index $l$ can be matched with up to $N$ waveguides, i.e., $\left|\mu\left(l\right)\right|\leqslant N$. 
(ii) Each waveguide $n$ can match (activate) at most $L$ pinching antennas, i.e., $\mu\left(n\right)\ne\emptyset$ and $1\leqslant \left|\mu\left(n\right)\right|\leqslant L$.  
(iii) The matching is bidirectional consistency, i.e., $l \in \mu\left(n\right)$ if and only if $n \in \mu\left(l\right)$. 

From definitions, the binary activation indicator is given by $a_{l,n}(\mu) = 1$ if $\left(l,n\right)$ is a matched waveguide-pinching antenna pair under $\mu$, 
and  $a_{l,n}(\mu) = 0$ otherwise. 
The above waveguide-pinching antenna matching $\mu$ is a two-sided many-to-many matching game. 
From problem (P0), the preference value of  waveguide $n$ over matching $\mu$ is defined as\footnote{
Based on the proposed matching algorithm, each waveguide seeks to match pinching antennas that reduce path losses for as many users as possible 
while concurrently alleviating multi-user interference, leading to activation patterns that naturally adapt to user distributions.
}
\begin{equation}
    U_{n}(\mu)= - \sum_{k \in \mathcal{K}}\left|w_{n,k}(\mu)\right|^2, ~ \forall n \in \mathcal{N}. 
\end{equation}
Similarly, the preference value of candidate pinching antenna location $l$ over matching $\mu$ is given by
\begin{equation}
    U_{l}(\mu)  = - \sum_{n\in\mathcal{N}} a_{l,n}(\mu) \sum_{k \in \mathcal{K}}\left|d_{n,k}(\mu)\right|^2, ~ \forall l \in \mathcal{L}, 
\end{equation}
where $\left|d_{n,k}\left(\mu\right)\right|^2 = \frac{1}{L_{n}^{\mathrm{s}}} \left|w_{n,k}\left(\mu\right)\right|^2$ denotes the equally radiated power of each pinching antenna over waveguide $n$. 
$\mathbf{W}(\mu)=[w_{n,k}(\mu)]\in\mathbb{C}^{N\times K}$ denotes the the optimal transmit beamforming for the given matching state $\mu$ and the corresponding pinching antenna activation 
$\mathbf{A}\left(\mu\right)$, 
which is obtained by solving the following convex SOC programming (SOCP):
\begin{equation*}
        \mathbf{W}\left(\mu\right) = \mathop{\arg\min}_{\mathbf{W}} ~ \left\Vert \mathbf{W}\right\Vert_{F} ^{2},
        \quad \text{s.t.} ~ \text{\eqref{P2_SINR}.}
\end{equation*}
Therefore, the minimum transmit power to guarantee SINR requirements based on $\mu$ can be evaluated.
Based on the KKT theory \cite{BF_Emil}, the optimal beamforming vector is given by 
$\mathbf{w}_{k}^{*}\left(\mu\right)=\sqrt{p_{k}^{*}\left(\mu\right)}\widetilde{\mathbf{w}}_{k}^{*}\left(\mu\right)$
with beamforming direction:
\begin{equation}\label{optimal_W_matching}
    \widetilde{\mathbf{w}}_{k}^{*}\left(\mu\right)\!=\!
    \frac{\left(\mathbf{I}_{N}\!+\!\sum\limits_{i=1}^{K}\frac{\lambda_{i}}{\sigma^{2}}
    \!\widetilde{\mathbf{h}}_{i}\!\left(\mu\right)\!\widetilde{\mathbf{h}}_{i}^{H}\!\left(\mu\right)\right)^{-1}\!
    \widetilde{\mathbf{h}}_{k}\!\left(\mu\right)}{\left\Vert\left(\mathbf{I}_{N}\!+\!\sum\limits_{i=1}^{K}\frac{\lambda_{i}}{\sigma^{2}}\!
    \mathbf{\widetilde{h}}_{i}\!\left(\mu\right)\!\widetilde{\mathbf{h}}_{i}^{H}\!\left(\mu\right)\right)^{-1}\!\mathbf{\widetilde{h}}_{k}\!\left(\mu\right)\right\Vert}.
\end{equation}
$\widetilde{\mathbf{h}}_{k}\left(\mu\right)=\mathbf{h}_{k}\mathbf{G}\mathbf{A}\left(\mu\right)$ denotes the effective channel vector of user $k$, 
and the Lagrangian multiplier $\lambda_{k}$ can be computed using numerical methods, 
such as the interior-point method \cite{InteriorPointComplexity} relying on Newton iterations. 
Similar to \cite{BF_Emil}, the beamfoming power $p_{k}^{*}\left(\mu\right)$, $\forall k \in\mathcal{K}$, can be obtained by 
\begin{equation*}
    p_{k}^{*}\left(\mu\right)\!=\!\frac{1}{\sigma^{2}\gamma_{\min}}\!\left|\widetilde{\mathbf{h}}_{k}^{H}\!\left(\mu\right)\!\mathbf{\widetilde{\mathbf{w}}}_{k}^{*}\!
    \left(\mu\right)\right|^{2}\!-
    \!\sum_{k'\ne k}\frac{1}{\sigma^{2}}\left|\widetilde{\mathbf{h}}_{k}^{H}\!\left(\mu\right)\!\mathbf{\widetilde{\mathbf{w}}}_{k'}\!\left(\mu\right)\right|^{2}\!\!.
\end{equation*}
Typically, a many-to-many matching problem can be solved by Gale-Shapley algorithm (also known as \textit{deferred acceptance} algorithm). 
If agents' preference lists are fixed and independent, Gale-Shapley algorithm can converge to a stable matching, in the sense that none of the waveguides/pinching antennas could change their current matching state without degrading other agents' satisfactions. 
However, the formulated waveguide-pinching antenna matching does not meet the above assumptions, as analyzed below. 

\begin{definition}[Externalities]
    A many-to-many matching exhibits \textit{externalities} if agents' preferences dynamically depend on the matching decisions or states of other agents.
    \end{definition}

\begin{definition}[Substitutability]
    A matching satisfies \textit{substitutability} if an agent prefers to remain matched with another agent when any subset of its existing matches is removed.
\end{definition}

\begin{proposition}
    The discrete activation of pinching antennas constitutes a many-to-many matching with externalities and non-substitutable preferences. Therefore, the existence of a stable matching is not guaranteed.
    \end{proposition}
    
    \begin{proof}
    As shown in \eqref{optimal_W_matching}, the optimal transmit beamforming $\mathbf{W}(\mu)$ depends on matching state $\mu$, which changes if any $a_{l,n}$ is updated. Hence, the preferences of waveguides and pinching antennas are influenced by the matching decisions of others, leading to externalities.
    Moreover, to achieve signal enhancement and interference mitigation, a pinching antenna or waveguide's matching behavior may vary with the composition of its matched group, violating the substitutability condition.
    Therefore, the classical conditions that ensure the existence of stable matchings do not hold in this setting~\cite{matching_externalities,matching_substitutability}.
\end{proof}

\subsection{Proposed Low-Complexity Solution}
We resort to swap matching theory to overcome the externalities.
Note that conventional swap matching mostly improves individual preferences of agents by searching for beneficial matching swaps.
However, since agents act selfishly, they may refuse a swap that harms agents' own utilities even though this swap improves the overall system performance (namely \textit{social welfares}). 
Hence, vanilla swap matching may converge to a pairwise stable matching that is far from locally or globally optimal solutions of problem (P0).
However, for solving problem (P0), a pinching antenna/waveguide is encouraged to increase individual power consumption if this helps reduce the total power consumption. 
Motivated by this, we develop a many-to-many matching algorithm that enables welfare-improving swap operations.  
This allows pinching antennas and waveguides to cooperatively minimize the total transmit power in problem (P0) whilst ensuring users' rate requirements.   
Compared to vanilla swap matching, the proposed algorithm guarantees both pairwise stability and local optimality.

\subsubsection{Welfare-Driven Matching Algorithm}
We begin by defining swap matchings and swap-blocking pairs, and then propose 
the concept of welfare-blocking pairs.

\begin{algorithm}[!t]
    \caption{Many-to-Many Matching Based Beamforming}
    \label{alg:swap_matching}
    \begin{algorithmic}[1]
    \REQUIRE Pinching antenna indices $\mathcal{L}$,  waveguides $\mathcal{N}$, $\mathbf{H}$, $\mathbf{G}$.
    \STATE Initialize $\mathbf{A}\left(\mu\right)=0$.
    \STATE Match each waveguide $n$ with its most preferred activated location $l$, and set $a_{l,n}(\mu)=1$.
    \item[] {//*~\textit{{Search for welfare-blocking pairs}}}
    \REPEAT
    \FOR{each unmatched pair $(l,n)$ in $\mu$}
        \STATE If $U\!\left(\mu_{\emptyset,n}^{l,\emptyset}\right)\!>\!U\!\left(\mu\right)$, update $\mu \!\leftarrow \!\mu_{\emptyset,n}^{l,\emptyset}$. {//~\textit{Add a match}} 
    \ENDFOR
    \FOR{each matched pair $(l,n)$ in $\mu$}
        \FOR{each unmatched pair $(l',n)$, $l' \ne l$}
            \STATE If $U\left(\mu_{l,n}^{l',\emptyset}\right)\!>\!U\left(\mu\right)$, update $\mu \!\leftarrow\! \mu_{l,n}^{l',\emptyset}$. {//~\textit{Replace}}
        \ENDFOR
        \FOR{each matched pair $(l',n') \ne (l,n)$}
            \STATE If $U\!\left(\mu_{l,n}^{l',n'}\right)\!\!>\!\!U\!\left(\mu\right)$, update $\mu \!\leftarrow\! \mu_{l,n}^{l',n'}$. {//~\textit{Exchange}}
        \ENDFOR
    \ENDFOR
    \UNTIL{no globally swap-blocking pairs remain.}
    \ENSURE Many-to-many matching $\mu$, pinching antenna activation $\mathbf{A}(\mu)$, and transmit beamforming $\mathbf{W}(\mu)$.
    \end{algorithmic}
\end{algorithm}

\begin{definition}[Matching Swap]
    A swap matching $\mu_{l,n}^{l',n'}$ swaps two existing pairs $(l,n)$, $(l',n')$ in matching $\mu$ into two 
    new pairs $(l,n')$, $(l',n)$.
\end{definition}

\begin{definition}[Vanilla Swap-Blocking Pair]
    A swap-blocking pair $(l,l',n,n')$ forms if
    \begin{enumerate}
        \item [(i)] For each agent $i\in(l,l',n,n')$, $U_{i}\left(\mu_{l,n}^{l',n'}\right) \geqslant U_{i}(\mu)$, i.e., individual utility is not decreased by $\mu_{l,n}^{l',n'}$. 
        \item [(ii)] For at least one agent $i\in(l,l',n,n')$,  individual utility is strictly improved by $\mu_{l,n}^{l',n'}$, i.e., $U_{i}\left(\mu_{l,n}^{l',n'}\right) > U_{i}(\mu)$.
    \end{enumerate}
\end{definition}

\begin{definition}[Welfare-Blocking Pair]
    A welfare-blocking pair $(l,l',n,n')$ exists if 
    (i) swap $\mu_{l,n}^{l',n'}$ is feasible, and 
    (ii) the social-welfare utility of all agents can be improved by swap $\mu_{l,n}^{l',n'}$, i.e., $U\left(\mu_{l,n}^{l',n'}\right) > U\left(\mu\right)$.
\end{definition}

We identify potential welfare-blocking pairs by examining following three types of swap operations, 
which allows the pinch/waveguide to add, replace, or exchange their matching: 
\begin{enumerate}
        \item[(i)] \textit{Add}: Add a new match $(l,n)$ into $\mu$ via a swap operation $\mu_{\emptyset,n}^{l,\emptyset}$. 
        Here, $\emptyset$ denotes an empty value. 
        Notations $(\emptyset,n)$ and $(l,\emptyset)$ are symbolic representations that do not alter the current matching states of waveguide $n$ and pinching antenna $l$, 
        which are purely introduced for convenience. 
        \item[(ii)] \textit{Replace}: Replace an existing match $(l,n)$ with a new match $(l',n)$ via the swap $\mu_{l,n}^{l',\emptyset}$, thereby removing $(l,n)$ from current matching.
        \item[(ii)] \textit{Exchange}: Exchange two existing matches $(l,n)$ and $(l',n')$ 
        to form new matches $(l,n')$ and $(l',n)$  via $\mu_{l,n}^{l',n'}$.
\end{enumerate}
The welfare-driven many-to-many matching game is summarized in \textbf{Algorithm \ref{alg:swap_matching}}. 

\subsubsection{Performance and Complexity Analysis}
To analyze the convergence behaviors of the developed many-to-many matching algorithm, we define the pairwise stability (also known as exchange stability), 
which ensures that no waveguide-pinching antenna pair desires to deviate from the current matching.

\begin{definition}[Welfare-Based Pairwise Stability]
    A matching $\mu$ is welfare-based pairwise stable if no welfare-blocking pair remains, i.e., 
    no feasible swap can improve the total utility.
\end{definition}

The convergence of \textbf{Algorithm \ref{alg:swap_matching}} is analyzed as follows.
\begin{theorem}\label{theorem:matching_conv}
    The proposed welfare-driven matching algorithm converges in finite steps to a welfare-based pairwise stable matching $\mu^*$ 
    and achieves a locally optimal solution of (P0).
    
    \begin{proof}
        At each iteration, a feasible swap is performed only if it leads to a strict increase in the system utility, i.e., 
        $U(\mu^{(t+1)}) > U(\mu^{(t)})$.
        Since each accepted swap strictly improves $U(\mu)$, previously visited matchings are never revisited, and cycles are avoided. 
        Moreover, the number of feasible matchings is finite, and the utility (e.g., the negative total transmit power) is upper bounded due to users' rate constraints. Hence, the algorithm must terminate in a finite number of iterations.
        
        Upon termination, there exists no feasible swap that can further increase the utility. Therefore, the final matching $\mu^*$ contains no welfare-improving blocking pairs and satisfies welfare-based pairwise stability. 
        Furthermore, since all feasible swaps are exhausted and none can improve $U(\mu^*)$, the matching is locally optimal within the 
        swap neighborhood. 
        Given that $\mathbf{W}(\mu^*)$ is the optimal transmit beamforming under the fixed matching $\mu^*$, no neighboring pair $(\mu', \mathbf{W}')$ can further reduce the total transmit power. Thus, the joint solution $(\mu^*, \mathbf{W}^*(\mu^*))$ is locally optimal.
        This ends the proof.
    \end{proof}
\end{theorem}

The computational complexity of \textbf{Algorithm \ref{alg:swap_matching}} can be analyzed as follows. 
To evaluate the utility for each matching, the time complexity for solving the optimal $\mathbf{W}^{*}$ is given by 
$\mathcal{O}\left(N^{3} K^{3}\right)$ based on interior point method \cite{InteriorPointComplexity,InteriorPoint_Iteration}. 
During initialization, $M$ new matches need to be evaluated (line 2 of Algorithm \ref{alg:swap_matching}), and the required time complexity is $\mathcal{O}\left(MN^{3}K^{3}\right)$.
Since there are at most $M$ unmatched pairs, the worst-case complexity of line 4-line 6 is given by $\mathcal{O}\left(MN^{3}K^{3}\right)$. 
Moreover, the worst-case complexity of line 8-line 10 and line 11-line 13 can be given by $\mathcal{O}\left((L-1)N^{3}K^{3}\right)$ and $\mathcal{O}\left((M-1)N^{3}K^{3}\right)$, respectively.
Since the system contains at most $M$ matched pairs, the worst-case complexity of line 7-line 14 is $\mathcal{O}\left(M(L+M-2)N^{3}K^{3}\right)$. 
Hence, the worst-case complexity of \textbf{Algorithm \ref{alg:swap_matching}} is given by $\mathcal{O}\left(I_{\mathrm{match}}M(L+M)N^{3}K^{3}\right)$, 
where $I_{\mathrm{match}}$ denotes the number of outer iterations of swap matching.

\begin{figure}[!tp]
    \vspace{-1.2em}
    \centering
    \subfloat[]{\includegraphics[width=1.7in]{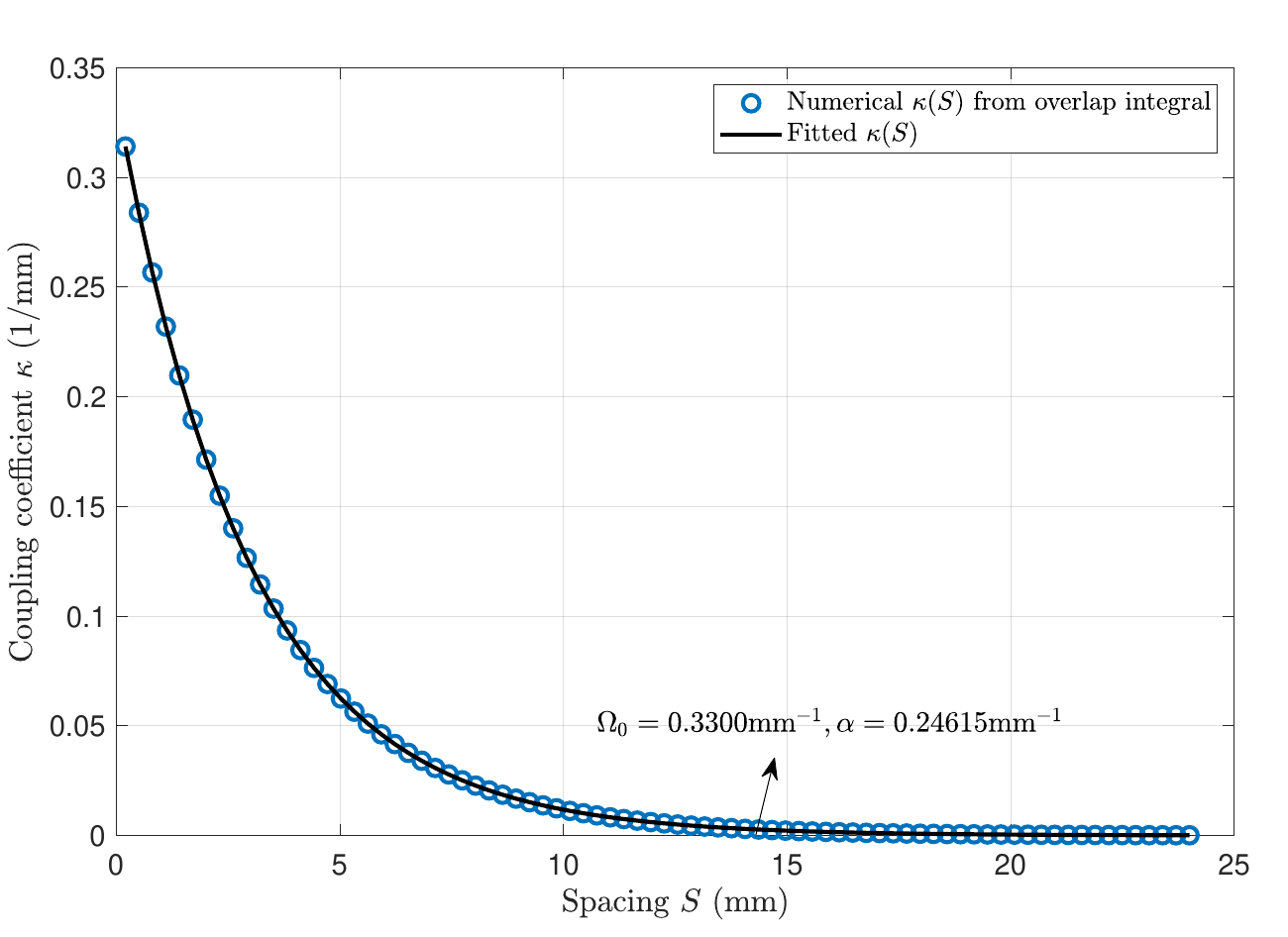}
   \label{fig:kappa_S}}
    \subfloat[]{\includegraphics[width=1.7in]{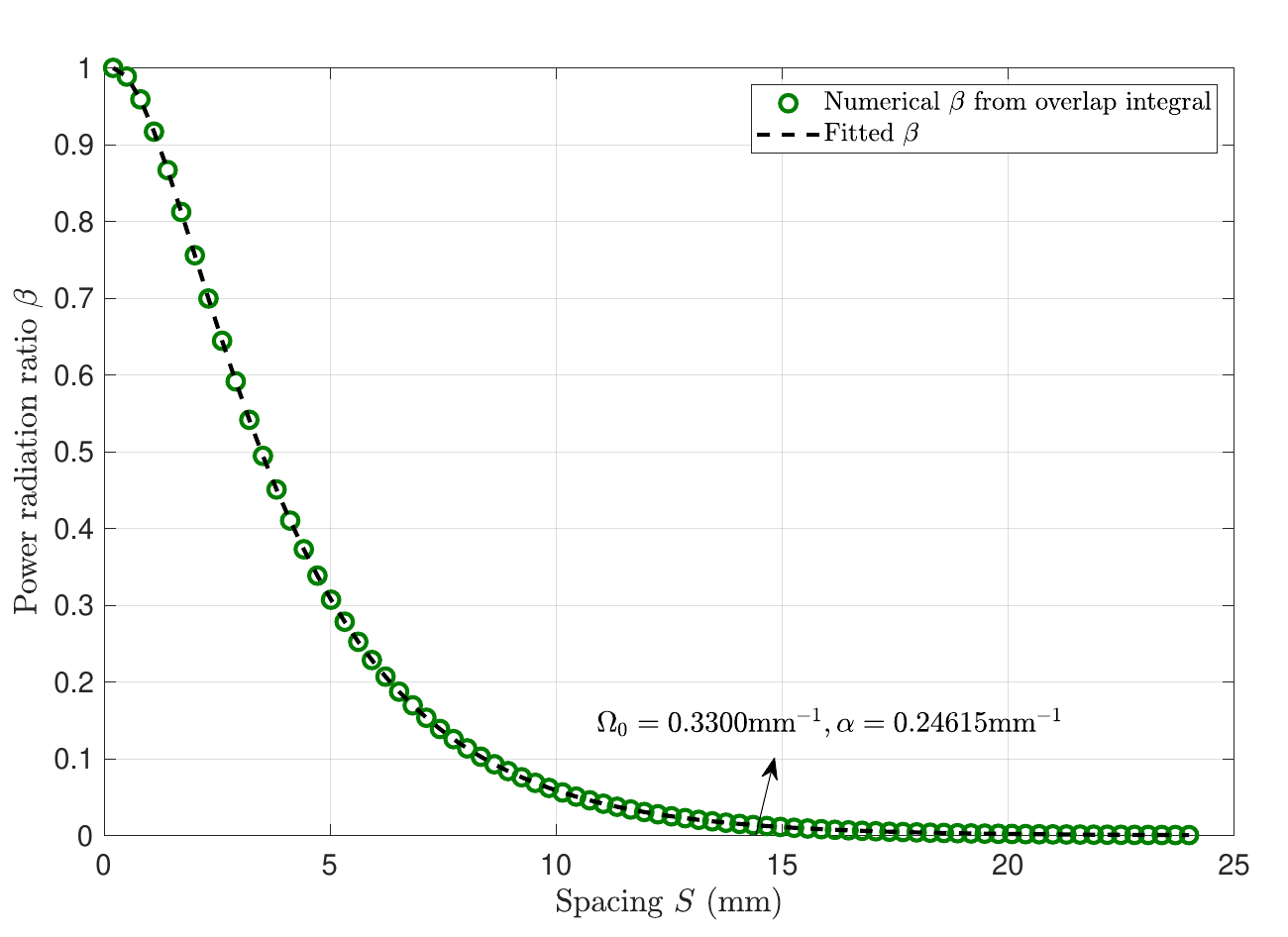}
    \label{fig:beta_S}}
    \caption{Coupling effects versus waveguide-antenna spacing.}
    \label{fig:fitting_S}
\end{figure}


\section{Simulation Results}
We present numerical results for both single-user and multi-user scenarios. 
The operating frequency is $f=15$ GHz, and the noise power is $\sigma^{2} = -80$ dBm. 
The refraction index is $n_{\mathrm{eff}}=1.4$. 
We consider a resource-limited scenario with $N=K$, which corresponds to the minimum RF chains (each connected with a waveguide) required to support $K$ simultaneous data streams.
Without special indications, we use the following setup in simulations. 
The number of users and RF chains/waveguides are $N=K=\{2,4\}$, the minimum SINR requirement is $\gamma_{\min}=20$ dB. 
The spatial ranges are  $S_{\mathrm{x}}=\{5,10,15,20,25,30\}$ and $S_{\mathrm{y}}=10$ meters.
The height of the BS is $H_{\rm PA}=5$ meters. 
For $N=2$, two waveguides are deployed parallel to the $x$-axis, and the feed points are given by $x_{n}^{\rm W} = 0$ and $y_{n}^{\rm W} = nS_{y}/2$, $n\in\{1,2\}$. 
For $N=4$, another two waveguides are deployed parallel to the $y$-axis with feed points $x_{n}^{\rm W} = (n-3)S_{x}/2$ and $y_{n}^{\rm W} = 0$, $n\in\{3,4\}$. 
Each waveguide has $L$ pre-mounted pinching antennas uniformly spaced along its span with intervals $S_x/L> \lambda_{f}$ or $S_y/L>\lambda_{f}$, 
and thus inter-antenna mutual coupling is neglected.

To verify waveguide's coupling effects, we compute the ground-truth coupling coefficient $\kappa(S)$ from the coupled-mode overlap integral using the evanescent field, 
and then estimate $(\Omega_{0},\alpha)$ via linear least-squares fitting of model \eqref{eq:coupling_coefficient}. 
A rectangular waveguide of cladding index $n_{\mathrm{clad}}=1.0$ and $2b=10$ mm is considered. 
The fitted parameters are obtained as
$\alpha = 0.24615~\mathrm{mm}^{-1}$ and $\Omega_{0}=0.3300~\mathrm{mm}^{-1}$.
As shown in Fig.~\ref{fig:fitting_S}, the fitted curve exhibits accurate approximation with the simulation data, 
confirming the effectiveness of model \eqref{eq:coupling_coefficient}.
Given an effective pinching antenna length $D^{\mathrm{PA}}=5\,\mathrm{mm}$, a minimum spacing $S_{\min}=0.1999\approx0.2\,\mathrm{mm}$ can be set,  
so that $\sin\left(\kappa(S_{\min})D^{\mathrm{PA}}\right)=\sin(\pi/2)=1$. 
For instance, when $L^{\mathrm{s}}=6$, the equal-power radiation spacings are given by $\{S_l\}_{l=1}^{6}=\{5.554,5.157,4.633,4.006,3.016,0.200\}$ (mm), which
monotonically decrease along the waveguide.

Two conventional MIMO systems deployed at the BS  with half-wavelength antenna spacing are considered as baselines.
\begin{itemize}
    \item \textbf{Massive MIMO}: A hybrid beamforming architecture \cite{SubConnectedHybridBF} is exploited, which equip $N$ RF chains, and each RF chain is connected to $L$ antennas via phase shifters. 
    Penalty-based method \cite{PDD_HybridBF} is employed to obtain the hybrid beamforming coefficients.
    \item \textbf{MIMO}: A conventional MIMO architecture is exploited, where each RF chain is connected with a single antenna. 
    Both the numbers of RF chains and antennas are equal to the number of users, i.e., $N=L=K$.
\end{itemize}

\subsection{Single-User Scenario}
In \textbf{Algorithm \ref{alg:branch_and_bound_SU}}, we consider the following schemes to determine the numbers of activated pinching antennas:
\begin{itemize}
    \item \textbf{BnB-Optimal}: Exhaustively search all combinations of possible $L_{n}^{\mathrm{s}}$ for each waveguide $n\in\mathcal{N}$, 
    and obtain the globally optimal solution by \textbf{Algorithm \ref{alg:branch_and_bound_SU}}.
    \item \textbf{BnB-Equal}: Force all waveguides to activate an equal number of pinching antennas, i.e., $L_{n}^{\mathrm{s}} = L^{\mathrm{s}}$, $\forall n$. 
    Hence, only a very small set $L^{\mathrm{s}} \in \{1,2,\dots,L\}$ needs to be enumerated in \textbf{Algorithm \ref{alg:branch_and_bound_SU}}, thus reducing complexity.
\end{itemize}

Fig. \ref{fig_SU_conv} demonstrates the convergence behaviors of the developed BnB algorithm for single-user scenario, where $N=2$, $L=12$, $S_x = S_y = 15$.
Specifically, the gaps between $\mathrm{GUB}$ and $\mathrm{GLB}$ converge to $0$ with only $10$ branching operations. 
This verifies that the developed BnB algorithm can efficiently search for the optimal solution of problem (P1-1) when the number of activated pinching antennas are fixed. 
Moreover, the BnB-Equal strategy achieves a similar performance with the BnB-Optimal strategy. 
This implies that for single-user scenario, using an equal number of pinching antennas at different waveguides can approximate the optimal performance while reducing the computational costs.

\begin{figure}[!htbp]
    \vspace{-0.8em}
    \centering
    \includegraphics[width=3.3in]{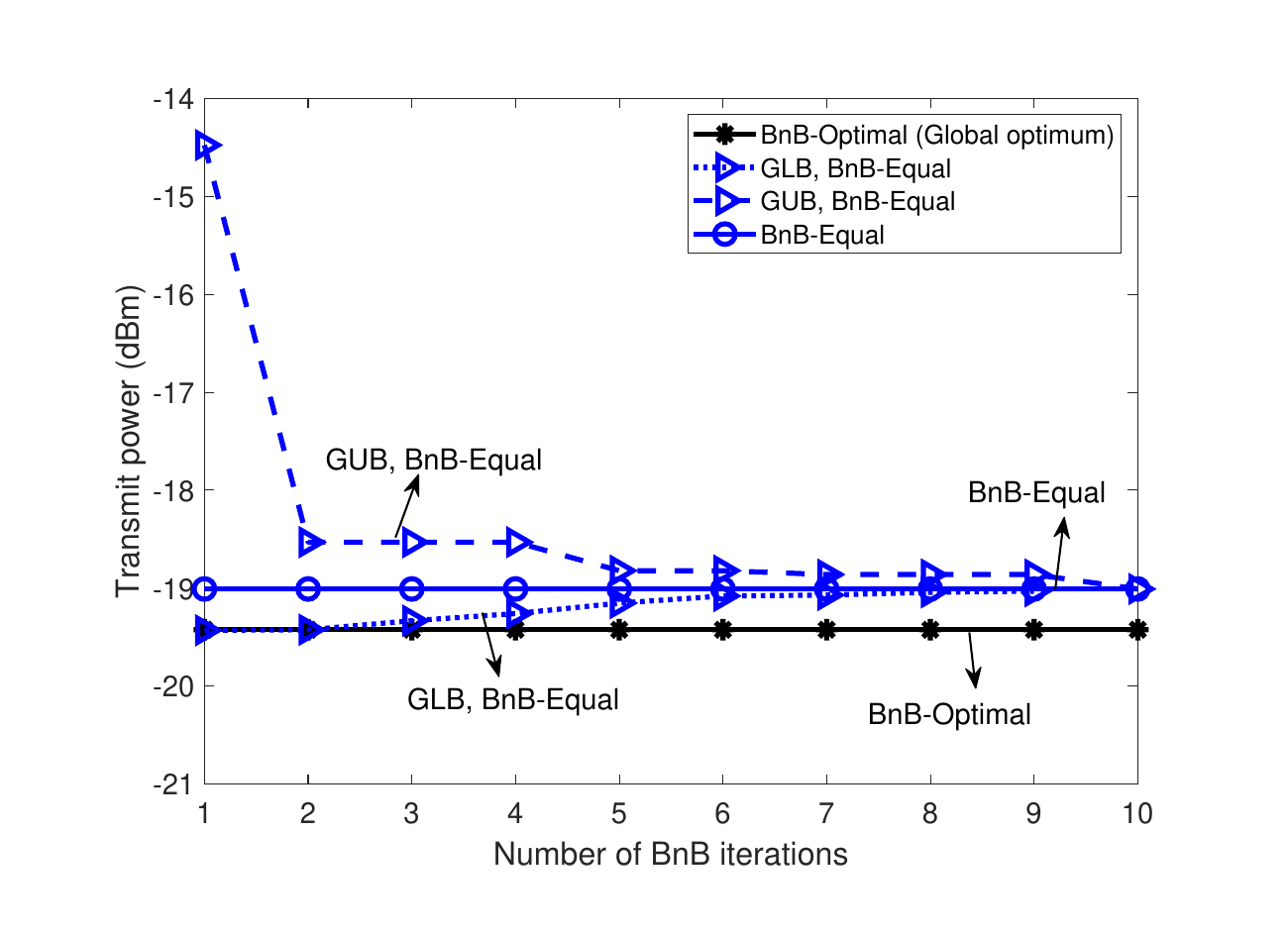}
    \caption{Convergence of BnB for single-user scenario.} 
    \label{fig_SU_conv}
\end{figure}

\begin{figure}[!htbp]
    \vspace{-0.4em}
    \centering
    \includegraphics[width=3.1in]{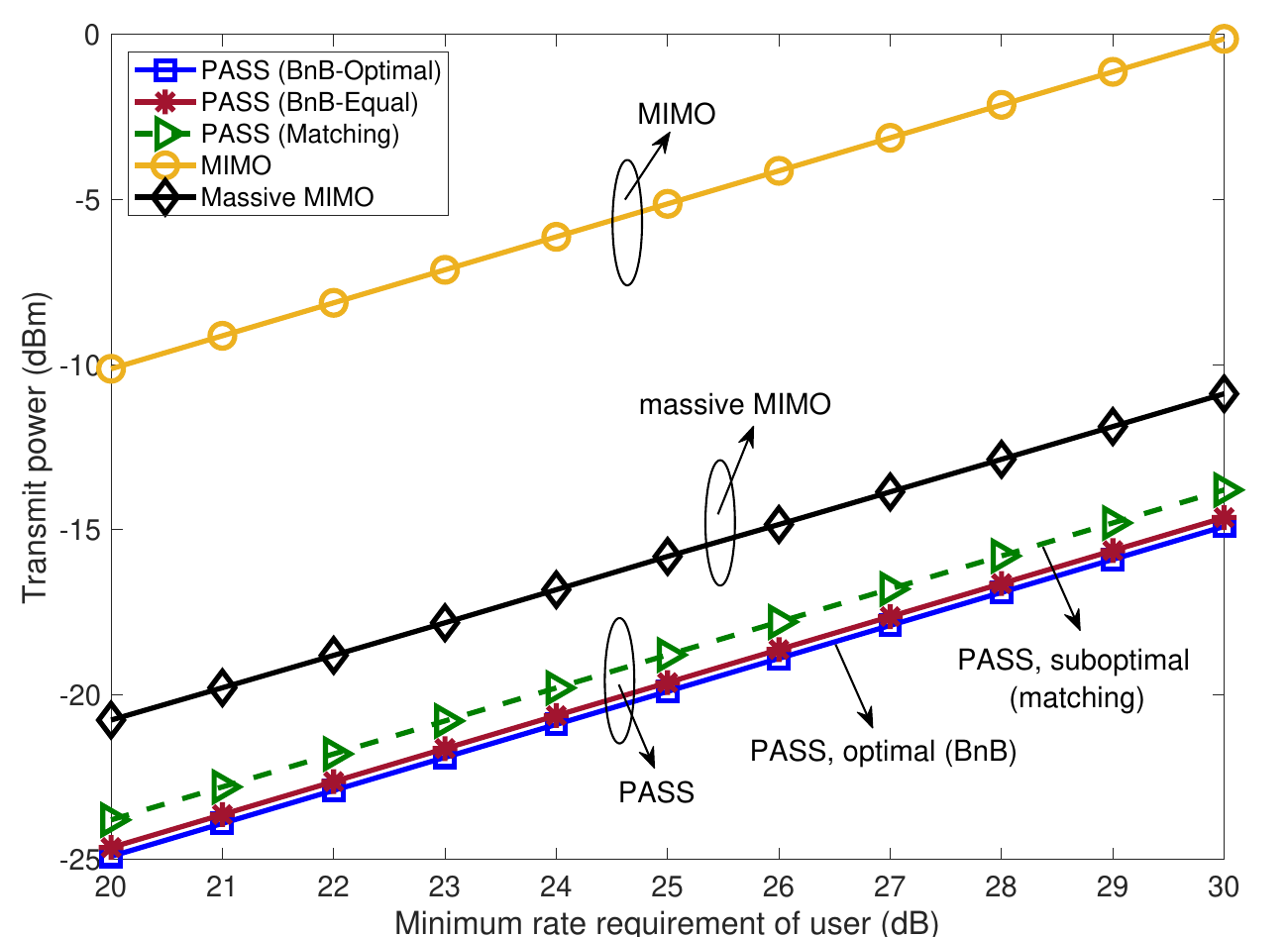}
    \caption{Performance comparisons in single-user scenario. }
    \label{fig_SU_gamma}
\end{figure}


Fig. \ref{fig_SU_gamma} compares the performance of single-user scenario and conventional single-user MIMO systems under different SINR requirement $\gamma_{\min}$, where 
$N=2$, $L=12$, $S_{x}=S_{y} = 15$.
The matching algorithm developed in Section IV is further extended to deal with the single-user scenario.
As shown in Fig. \ref{fig_SU_gamma}, conventional MIMO system requires the highest power consumption to satisfy user's SINR requirement. 
In comparison, massive MIMO system significantly reduces the power consumption by adopting the hybrid beamforming architecture that is more energy-efficient.  
By flexibly activating pinching antennas next to the user, PASS can reduce over $30\%$ power consumption compared to massive MIMO in single-user scenarios. 
Compared to BnB-Optimal, both BnB-Equal and matching algorithm realize neal-optimal performance.

\subsection{Multi-User Scenario}
\vspace{-0.2em}

\begin{figure}[!htbp]
    \vspace{-1.2em}
    \centering
    \includegraphics[width=3.1in]{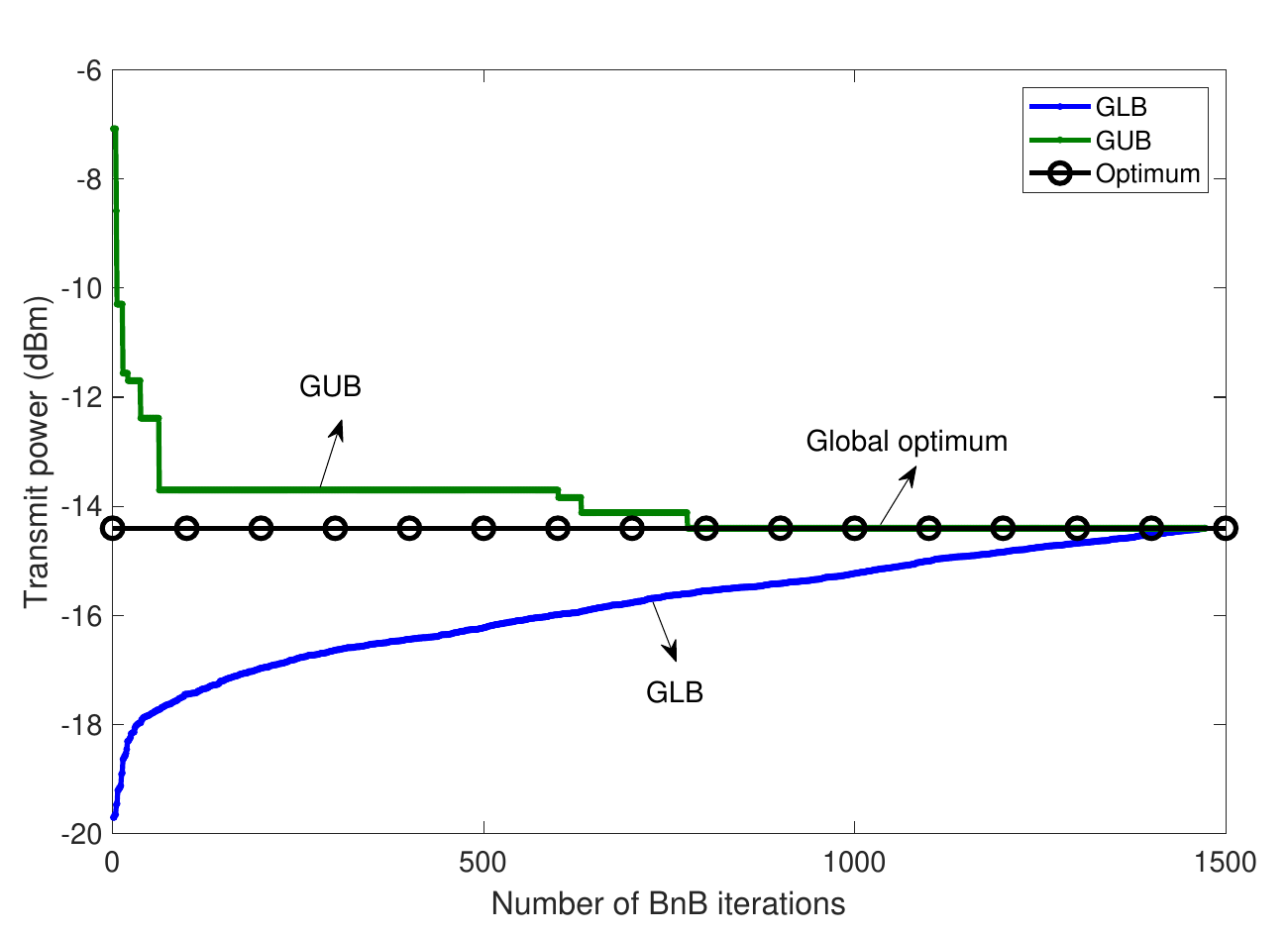}
    \caption{Convergence of BnB for multi-user scenario.} 
    \label{fig_BnBconv_MU}
\end{figure}

We evaluate the multi-user scenario performance in this part. 
Fig. \ref{fig_BnBconv_MU} demonstrates the convergence behavior of the proposed globally optimal BnB algorithm (see \textbf{Algorithm \ref{alg:branch_and_bound_MU}}), 
where $L=6$, $K=N=4$. 
As observed, the global lower bound increases monotonically, while the global upper bound decreases with each branching step, 
which confirms the validity of the convex relaxation and bounding functions.
The gap between $\mathrm{GUB}$ and $\mathrm{GLB}$ monotonically narrows and eventually approaches zero. 
This verifies that the proposed BnB algorithm guarantees convergence to the $\varepsilon$-optimal solution within a finite number of iterations, 
which is consistent with the theoretical analysis in \textbf{Theorem \ref{theorem:MU_BnBConverge}}.
Compared to exhaustive search, which requires evaluating $2^{NL} = 2^{24} \approx 1.6777 \times 10^{7}$ combinations of antenna activations and is computationally forbiditive, 
the proposed BnB algorithm achieves the global optimum within $2000$ iterations.
Nevertheless, the computational complexity remains high, motivating the development of more efficient low-complexity algorithm.

\begin{figure}[!htbp]
    \vspace{-0.6em}
    \centering
    \includegraphics[width=3.2in]{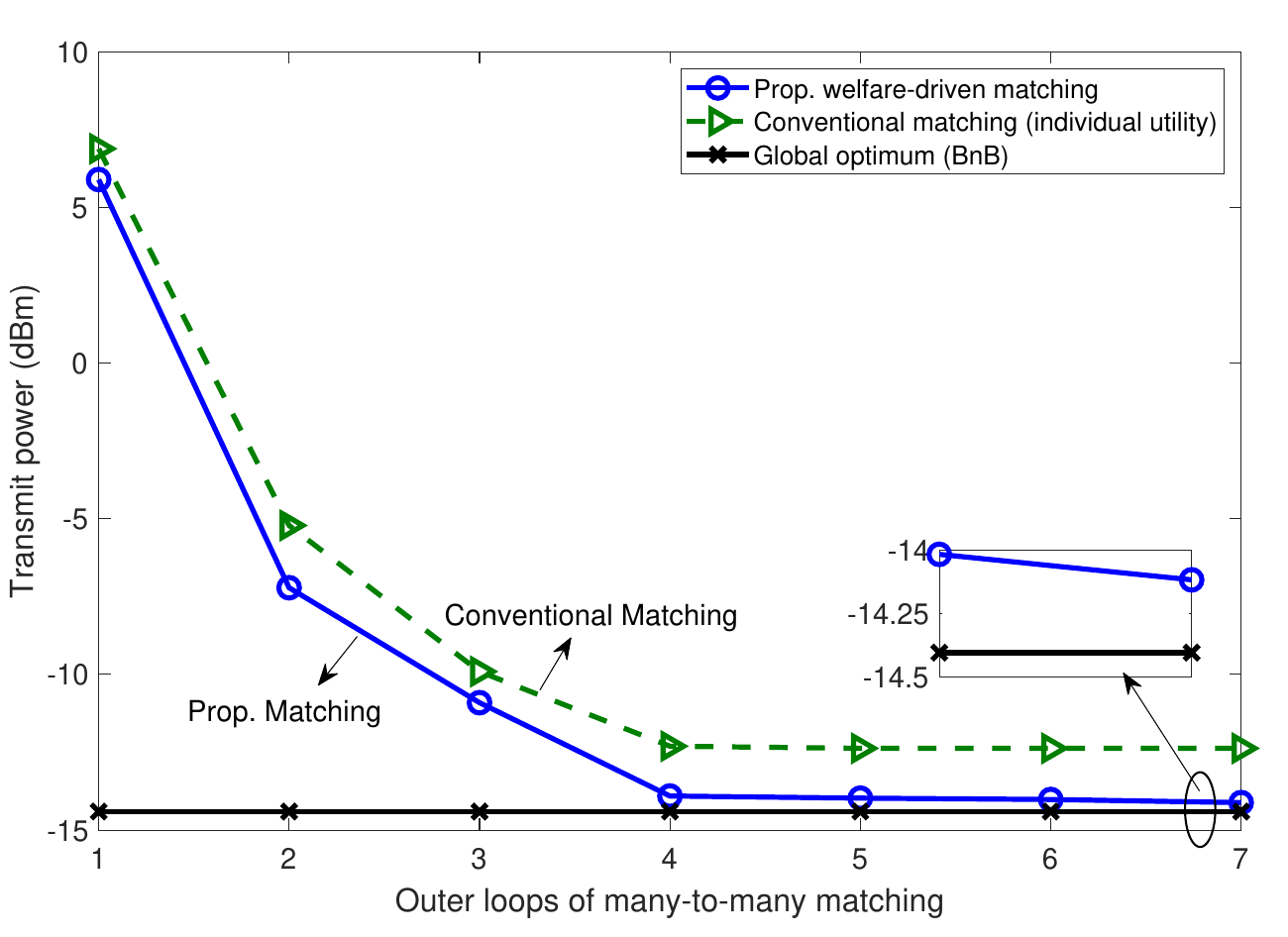}
    \caption{Convergence of many-to-many matching algorithm.} 
    \label{fig_MU_SMconv}
\end{figure}

We further evaluate the convergence behavior of the proposed  welfare-driven many-to-many matching algorithm. 
Fig. \ref{fig_MU_SMconv} exhibits the total transmit power versus the number of outer iterations in \textbf{Algorithm \ref{alg:swap_matching}}. 
It can be observed that the total transmit power decreases monotonically with swap operations 
and eventually converges to a stable value within approximately $I_{\mathrm{match}}=10$ outer loops, making it a computationally tractable solution suitable for practical implementation. 
The algorithm terminates once no further welfare-improving blocking pairs exist, thereby reaching a pairwise stable matching, aligned with \textbf{Theorem~\ref{theorem:matching_conv}}.
Notably, the proposed welfare-driven many-to-many matching realizes near-optimal performance with only marginal loss compared to the globally optimal BnB algorithm, 
while substantially reducing computational complexity.
In contrast, conventional swap matching based on individual utility fails to approach the global optimum, as the agents are not incentivized to contribute to the overall utility. 

\begin{figure}[!htbp]
    \vspace{-0.6em}
    \centering
    \includegraphics[width=3.2in]{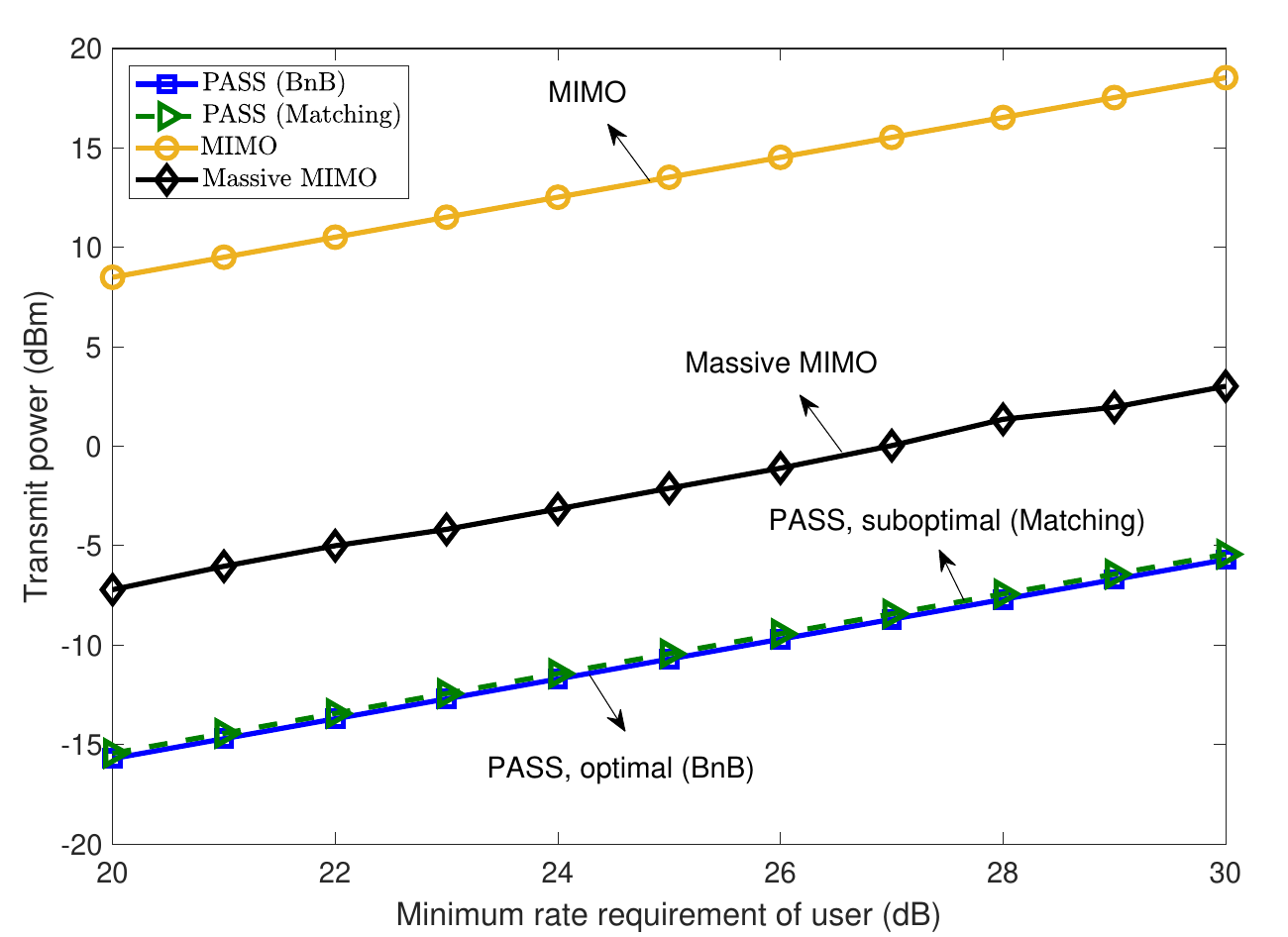}
    \caption{Transmit power versus $\gamma_{\min}$. }\label{fig_MU_gamma}
\end{figure}

Fig. \ref{fig_MU_gamma} compares the total transmit power of different architectures versus the minimum SINR requirement $\gamma_{\min}$, 
where $L=6$, $K=2$.
The transmit power consumption increases monotonically with the SINR requirement for all schemes, 
as higher rate requirements demand stronger signal power.
Among the compared architectures, conventional MIMO and massive MIMO consume significantly higher transmit power, particularly under high SINR requirements.
In contrast,  PASS with the globally optimal BnB algorithm achieves the lowest transmit power, 
reducing power consumption by over $22$ dBm and $7.5$ dBm compared to MIMO and massive MIMO, respectively. 
This means that PASS achieves over $99\%$ and $80\%$ power savings then MIMO and Massive MIMO, respectively, 
confirming its capabilities in adjusting large-scale path loss and achieving energy savings. 
Furthermore, the proposed welfare-driven matching algorithm achieves near-optimal performance despite different SINR requirements, 
which demonstrates its effectiveness and practicality.

\begin{figure}[!htbp]
    \vspace{-0.6em}
    \centering
    \includegraphics[width=3.2in]{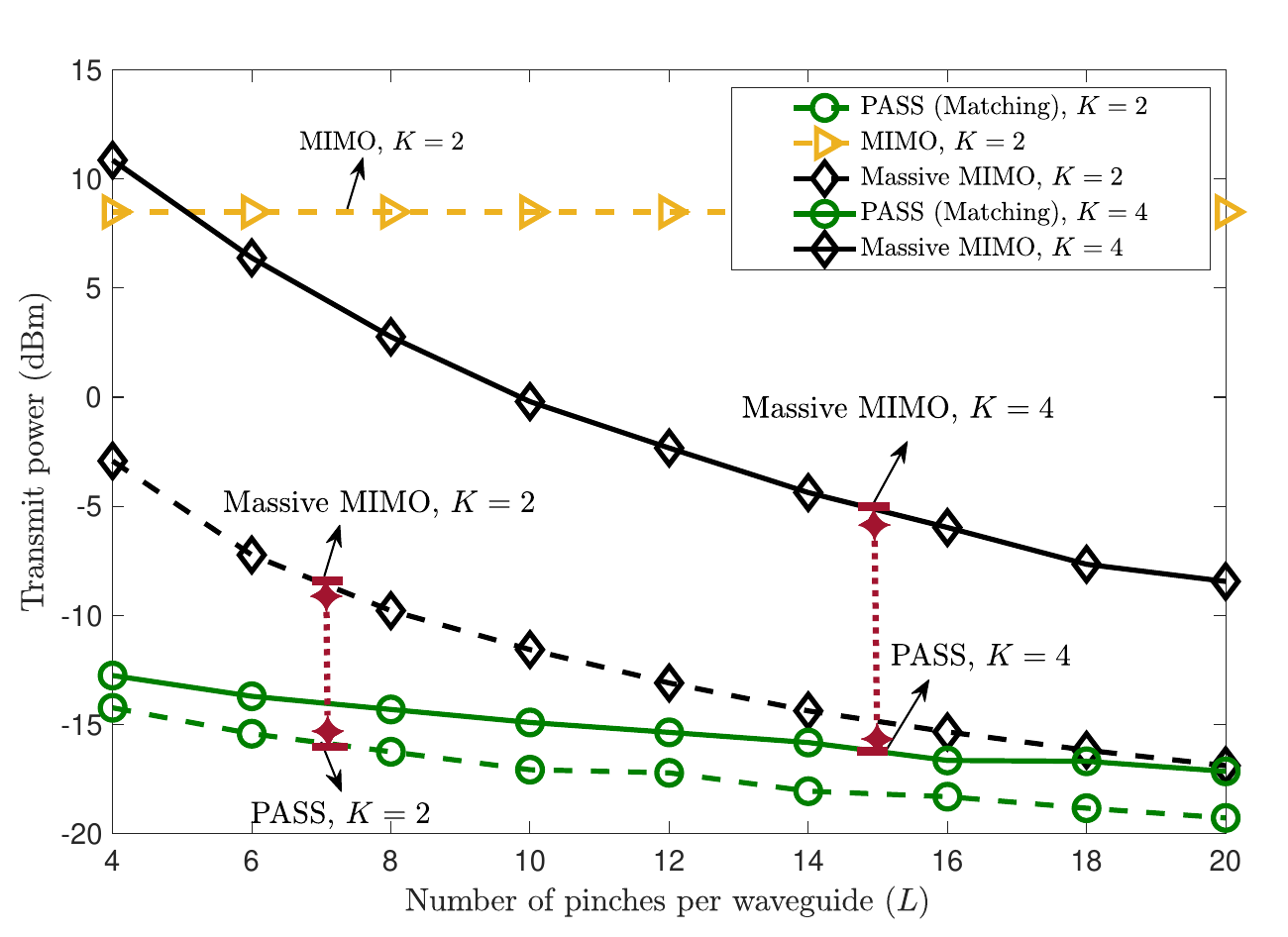}
    \caption{Transmit power versus $L$. }
    \label{fig_MU_L}
\end{figure}

Fig. \ref{fig_MU_L} illustrates the impact of the number of pinching antennas per waveguide $L$ 
on the transmit power of different architectures. 
It is observed that the transmit power required by both PASS and massive MIMO decreases with increasing $L$, 
as the additional pinching antennas provide greater spatial degrees of freedom. 
Compared to massive MIMO, PASS achieves a lower transmit power across all $L$, highlighting its reconfigurability and scalability advantages.
Moreover, the performance gain increases with the number of users. 
For instance, when $K=4$ and $L=14$, PASS reduces over $10$ dBm transmit power than massive MIMO. 
This confirms the capability of PASS to efficiently serve users distributed across different spatial regions through adaptive pinching antennas. 

\begin{figure}[!htbp]
    \vspace{-0.6em}
    \centering
    \includegraphics[width=3.2in]{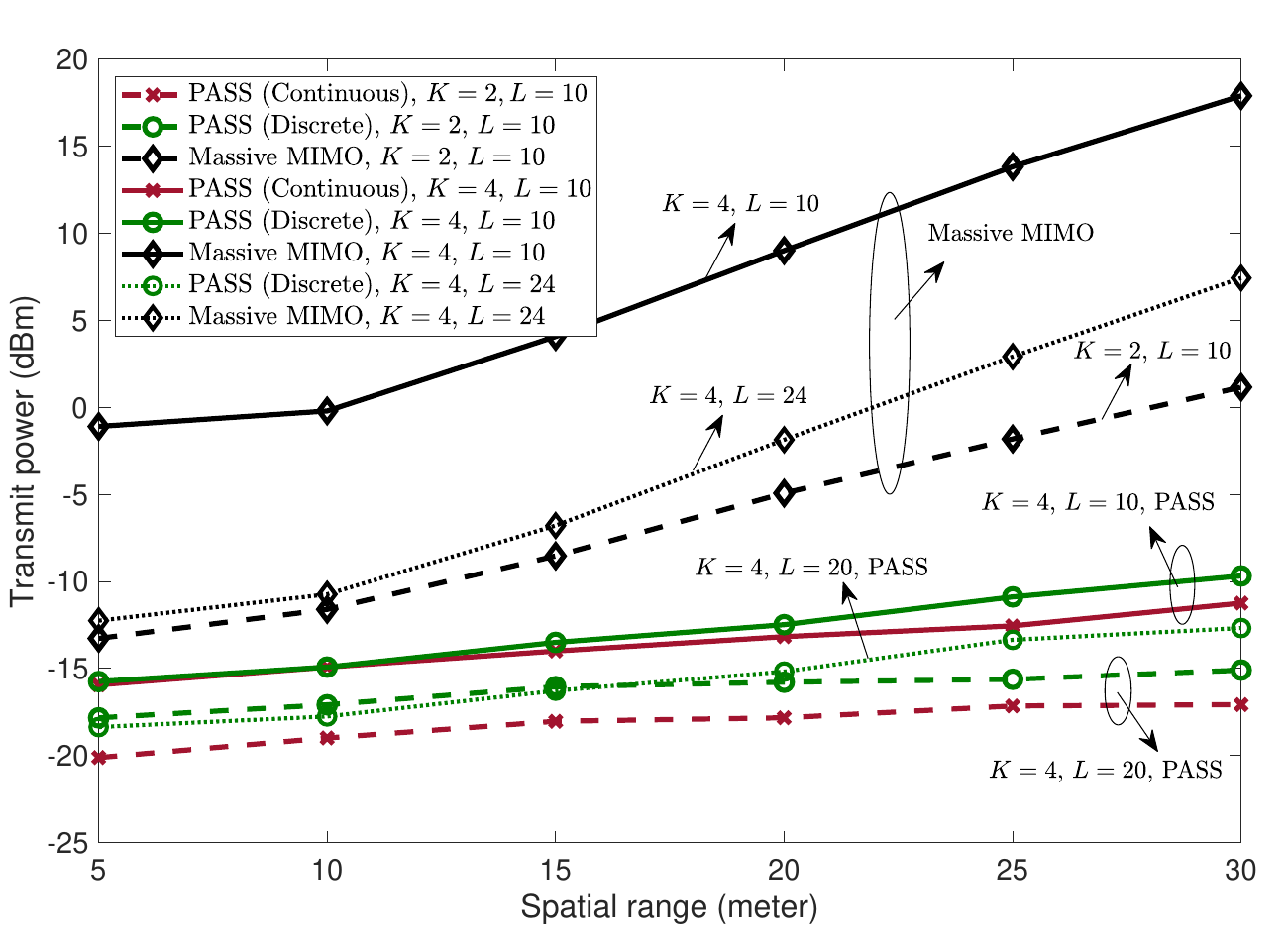}
    \caption{Transmit power versus spatial range $S_{\mathrm{x}}$.} 
    \label{fig_MU_S}
\end{figure}

Fig. \ref{fig_MU_S} presents the system performance of different architectures versus the spatial range $S_{\mathrm{x}}$, with $\gamma_{\min} = 20$ dB and $L=10$. 
The continuous activation is included as an additional benchmark, 
where each pinching antenna's location is flexibly adjusted in a continuous space through grid searching over $50$ grids. 
As the spatial range increases, the power consumption of conventional massive MIMO system significantly grows to compensate for larger path loss of users. 
In contrast, both discrete and continuous activations require slight increments in transmit power when spatial range increases. 
This is because pinching antennas can be activated near the users,  thereby maintaining communication quality over the coverage area.
While continuous activation leads to the best behaviors,  
the discrete structure provides a practical and cost-effective alternative with reduced implementation complexity, 
which suggests its viability and potential in practical applications. 
Furthermore, as the number of pinching antennas increases, the required transmit power effectively decreases, which confirms the scalability of the proposed algorithm.

\section{Conclusion}
A novel adjustable power radiation model for PASS has been proposed, which enables power radiation ratios to be flexibly adjusted by configuring spacing between pinching antennas and waveguides. 
The closed-form pinching antenna spacing arrangement to achieve equal-power radiation was derived for any arbitrary number of activated antennas.
Exploiting this in downlink PASS communications with practical discrete activation, pinching beamforming and transmit beamforming have been jointly optimized to minimize total transmit power, subject to users' 
SINR constraints. 
Globally optimal BnB algorithms have been proposed for both single-user and multi-user scenarios, with theoretical guarantees on convergence and optimality.
To reduce the computational complexity, a welfare-driven many-to-many matching algorithm was further proposed to obtain locally optimal and pairwise-stable solutions within polynomial-time complexity. 
Simulation results confirmed that PASS outperforms traditional multi-antenna architectures, particularly when the user number and the spatial range increase, and 
the proposed welfare-driven matching attains near-optimal performance with much lower complexity. 
These results highlight the potentials of PASS as a promising architecture for next-generation wireless systems.
Further research may explore advanced designs to support large-scale user deployments and mitigate dynamic blockages.

\appendices
\section{Proof of Proposition \ref{proposition:coupling_strength}}\label{proof:proposition:coupling_strength}
Note that \textbf{Proposition \ref{proposition:coupling_strength}} generalizes analytical results in \cite{Okamoto_waveguide}, and is aligned with experimental results in \cite{dynamic_coupling_spacing}. 
The analytical expressions of coupling coefficients for both rectangular and circular waveguides can be derived following \cite{Okamoto_waveguide}. 
\textit{First}, for a rectangular waveguide and a rectangular non-contact coupler (the pinching antenna) with an equal core width $2b$ (along the $y$-axis in Fig. \ref{fig_waveguide}),  
using the evanescent field in the cladding yields the analytical expression of the coupled-mode overlap integral \cite[Eq.~(4.91)]{Okamoto_waveguide} 
    \begin{equation}
    \label{eq:ok_rect_491}
    \kappa_{\mathrm{rect}}=\tfrac{\sqrt{2\Delta_{0}}}{b}\tfrac{k_0^2\alpha^{2}b^{4}}{\left(1+\alpha b\right)v^3}~e^{-\alpha(S-2b)},
    \end{equation}
    where $S$ is the center-to-center spacing (so $S-2b$ is the edge-to-edge gap), 
    $k_0$ is the transverse core wavenumber, $\alpha = \sqrt{4\pi^2/\lambda_{f}^2 (n_{\rm eff}^2 - n_{\rm clad}^2)-k_0^2}$ is the cladding decay constant along $z$-axis \cite[Eq. (2.48)]{Okamoto_waveguide}, 
    with propagation constant $\gamma_{0} = \sqrt{4\pi^2/\lambda_{f}^2n_{\rm eff}^2-k_0^2}$, 
    $\Delta_{0}$ denotes the relative index contrast to the core,  and $v=2\pi/\lambda ~ n_{\rm eff}b\sqrt{2\Delta_{0}}$ is the normalized frequency. 
    By setting $\Omega_{0} = \tfrac{\sqrt{2\Delta_{0}}}{b}\tfrac{k_0^2\alpha^{2}b^{4}}{\left(1+\alpha b\right)v^3} e^{-2\alpha b}$, 
\eqref{eq:ok_rect_491} reduces to \eqref{eq:coupling_coefficient}. 
\textit{Secondly}, for circular dielectric guides (core radius $b$, center separation $D$), 
the coupling coefficient formula can be derived based on the exterior guided field with modified Bessel functions, which leads to \cite[Eq. (4.121)]{Okamoto_waveguide}
\begin{equation}
    \label{eq:ok_circular_121}
    \kappa_{\mathrm{circ}}=\tfrac{\sqrt{\Delta_{0}}}{b}\tfrac{u^2}{K_1^2(w)v^3}\tfrac{\sqrt{\pi b}}{wS}~e^{-\alpha(S-2b)},
\end{equation}
where $v\!=\!2\pi/\lambda_{f} b \sqrt{n_{\rm eff}^2 \!-\! n_{\rm clad}^2}$, $u \!=\! b\sqrt{4\pi^2/\lambda^2n_{\mathrm{eff}}^2\!-\!\gamma_{0}^2}$, and $w \!=\! b\alpha$. 
Thus, $\kappa_{\rm circ}$ also reduces to (3).
This ends the proof.

\section{Proof of Lemma \ref{lemma:equalpower}}\label{proof:lemma:equalpower}
    We ignore the waveguide index $n$ here, and define $\delta_l \triangleq \sin\left(\kappa_{l} D^{\mathrm{PA}}\right)$. Let sequence \(\{\delta_l\}_{l=1}^{L_n^{\mathrm{s}}}\) satisfy the relationship
$\delta_l \prod_{i=1}^{l-1} \sqrt{1 - \delta_i^2} = \beta, ~ \forall l \in \{1, \dots, L^{\mathrm{s}}\}$.
By recursively solving this equation we have 
$\delta_l = \tfrac{\beta}{\sqrt{1 - \sum_{i=1}^{l-1} \delta_i^2}}$.
Substituting recursively yileds the general expression:
\begin{equation}\label{general_radiating_control}
    \delta_l \!=\! \tfrac{\beta}{\sqrt{1 \!-\! (l-1)\beta^2}}
    \overset{(a)}{=} 
    \tfrac{1}{\sqrt{L^{\mathrm{s}} \!-\! (l - 1)}}, ~ \text{for } l = 1, 2, \dots, L^{\mathrm{s}},
\end{equation}
where (a) is obtained by setting $\beta = \tfrac{1}{\sqrt{L^{\mathrm{s}}}}$. 
Gathering the activated antennas from all the candidate antennas, the closed-form solution of $\delta_{l}$ is written as
\begin{equation} 
    \delta_{l}\!=\!\tfrac{\beta}{\sqrt{1\!-\!\rho_{l}\beta^{2}}}
    \!=\! \tfrac{1}{\sqrt{L^{\mathrm{s}} \!-\! \rho_{l}}}, ~
    \text{for $l = 1,2,...,L$}.
\end{equation}
From the definition of $\delta_{l}$ it follows that
$\kappa_l=\tfrac{\arcsin(\delta_l)}{D^{\mathrm{PA}}}$ and 
$S_l=\tfrac{1}{\alpha}\,\ln\!\left(\tfrac{\Omega_{0}}{\kappa_l}\right)$, $l=1,\dots,L^{\mathrm{s}}$,
which completes the proof.

\section{Proof of Lemma \ref{lemma:GLB_GUB_diff}}\label{proof:lemma:GLB_GUB_diff}
Let $\mathbf{x}_{\mathrm{c}}=\big\{\mathbf{A}_{\mathrm{c}},\mathbf{D}_{\mathrm{c}},\mathbf{Z}_{\mathrm{c}}\big\}$  and 
$\mathbf{x}_{\mathrm{GUB}}=\big\{\mathbf{A}_{\mathrm{proj}},\allowbreak\mathbf{D}_{\mathrm{proj}},\mathbf{Z}_{\mathrm{proj}}\big\}$ 
denote the vectorized solutions that achieve $\mathrm{GLB}$ and $\mathrm{GUB}$.
From Lagrange mean-value theorem, we have
\begin{equation}\label{GUB_GLB_diff}
    \begin{split}
        &\mathrm{GUB} \!-\! \mathrm{GLB} \!=\! F\!\left(\mathbf{x}_{\mathrm{GUB}}\right) \!-\! F\!\left(\mathbf{x}_{\mathrm{c}}\right) 
        \!=\! \nabla_{\mathbf{x}} F^{T}\!\left(\mathbf{x}\right)\left(\mathbf{x}_{\mathrm{GUB}} \!-\! \mathbf{x}_{\mathrm{c}}\right) \\
    &\overset{(a)}{\leqslant}\!\! \left\Vert \nabla_{\mathbf{x}} F\left(\mathbf{x}\right) \right\Vert \!
    \left\Vert \mathbf{x}_{\mathrm{GUB}} - \mathbf{x}_{\mathrm{c}} \right\Vert
    \!\overset{(b)}{=}\! \sqrt{2MP_0} \left\Vert \mathbf{x}_{\mathrm{GUB}} - \mathbf{x}_{\mathrm{c}} \right\Vert,
    \end{split}
\end{equation}
where $\mathbf{x} \in \left\{\mathbf{y} \mid \mathbf{y} = t \mathbf{x}_{\mathrm{GUB}} + (1-t) \mathbf{x}_{\mathrm{c}}, ~ t\in\left[0,1\right]\right\}$. 
Inequality (a) comes from Cauchy-Schwarz inequality, 
and inequality (b) results from the fact that the $l_{2}$-norm of the gradient 
$\nabla_{\mathbf{x}}F\left(\mathbf{x}\right)=\left[\mathbf{0}_{M\times1},\mathbf{0}_{NK\times1},2\mathbf{z}_{1}^{T},2\mathbf{z}_{2}^{T},\dots,2\mathbf{z}_{K}^{T}\right]^{T}$ 
is given by $\left\Vert \nabla_{\mathbf{x}}F\left(\mathbf{x}\right)\right\Vert =\sqrt{2}\left\Vert \mathbf{Z}\right\Vert _{F}\leqslant\sqrt{2MP_{0}}$.  
Combining $\left\Vert \mathbf{x}_{\mathrm{GUB}} - \mathbf{x}_{\mathrm{c}} \right\Vert \leqslant \sqrt{B} \phi_{\max}$ and 
\eqref{GUB_GLB_diff}, we have 
\begin{equation}
    \mathrm{GUB}-\mathrm{GLB} \leqslant \sqrt{2MP_{0}B}\phi_{\max}.
\end{equation}
Hence, by selecting $\phi_{\max} \leqslant \varepsilon/\sqrt{2MP_{0}B}$, we have $\mathrm{GUB}-\mathrm{GLB} \leqslant \varepsilon$. 
This completes the proof.

\section{Proof of Theorem \ref{theorem:MU_BnBConverge}}\label{proof:theorem:MU_BnBConverge}
We first demonstrate that the proposed method satisfies classic BnB convergence conditions \cite{BnB_convergence}. 
Specifically, $GLB$ and $GUB$ of \textbf{Algorithm \ref{alg:branch_and_bound_MU}} converge in a finite number of iterations if the following conditions hold \cite{BnB_convergence}:
    \begin{itemize}
        \item[1)] \textbf{Bound validity}: The upper and local bounds become tight as the length of boxes shrinks to a point. 
        \item[2)] \textbf{Exhaustiveness}: The length of branched box decreases to zero as the number of iterations approaches infinity.
        \item[3)] \textbf{Bound convergence}: The gap $\mathrm{GUB}-f_{\mathrm{true}}^{*}$ vanishes as the maximum edge length approaches zero.  
    \end{itemize}
    First, since the equalities in constraints \eqref{McCormickE_mu} hold true when $\overline{\mathbf{b}}=\underline{\mathbf{b}}$, 
    McCormick envelope shrinks to bilinear constraints as the length of boxes shrinks to a point. 
    Hence, \textbf{Algorithm \ref{alg:branch_and_bound_MU}} satisfies condition 1).
    Moreover, condition 2) holds true based on the employed box selection and branching rules.
    From definitions of $\mathrm{GUB}$ and $\mathrm{GLB}$, we have $\mathrm{GLB} \leqslant f^{*} \leqslant \mathrm{GUB}$, which implies that
        $0 \leqslant \mathrm{GUB} - f_{\mathrm{true}}^{*} \leqslant \mathrm{GUB}-\mathrm{GLB}$.
    Combining \textbf{Lemma \ref{lemma:GLB_GUB_diff}}, when $\phi_{\max}\leqslant \varepsilon/\sqrt{2MP_{0}B}$, we have 
    \begin{equation}\label{convergence_BnB}
        \mathrm{GUB} - f_{\mathrm{true}}^{*} \leqslant \mathrm{GUB}-\mathrm{GLB} \leqslant \varepsilon, ~ \forall \varepsilon \geqslant 0.   
    \end{equation}
    Hence, the bounding procedure converges as $\phi_{\max} \rightarrow 0$, and condition 3) is also satisfied.
    Since the algorithm terminates with $f^* = \mathrm{GUB}$, inequality \eqref{convergence_BnB} guarantees that \( f^* \leq f_{\mathrm{true}}^{*} + \varepsilon \). 
    This ends the the proof.

\section{Proof of Theorem \ref{theorem:BnB_complexity}}\label{proof:theorem:BnB_complexity}
For binary discrete variables $\mathbf{A}$, it is obvious that at most $2^{M}$ partitions with $M$ tree depths need to be searched. 
For continuous variables, at most $T_{\mathrm{c}}$ (branching) iterations with $N_{\mathrm{c}}$ tree depths are required, as analyzed below. 
In the worst case, assume that \textbf{Algorithm \ref{alg:branch_and_bound_MU}} converges at the $T$-th iteration, and 
the BnB tree depth is $N_{\rm tree}=M+N_{\rm c}$. 
The edge lengths of the best node that reaches $GUB$ are denoted by 
$\phi_{1}, \phi_{2}, \dots, \phi_{B}$. 
Combining the branch rule and \textbf{Lemma \ref{lemma:GLB_GUB_diff}}, 
before each edge $i$ performing its last partition, its length is $2\phi_{i}$ and satisfies
\begin{equation}\label{edge_length_bound}
    2\phi_{i} \geqslant \xi, ~ \forall i = 1,2, \dots, B.
\end{equation}
Hence, the volume of the branching node  over $B_{\rm c} \triangleq 2NK$ continuous variables at the $N_{\rm tr}$-th depth level satisfies
    $\psi_{N_{\rm tr}}=\frac{\psi_{\mathrm{vol}}}{2^{N_{\mathrm{c}}}}=\prod_{i=1}^{B_{\rm c}}\phi_{i}
    \overset{\eqref{edge_length_bound}}{\geqslant}\left(\tfrac{\xi}{2}\right)^{B_{\rm c}}.$
After rearrangement and combining $B=B_{\rm c} + M$, the tree depth can be bounded by 
\begin{equation*}
        N_{\mathrm{c}}\!\leqslant\! \!\left\lceil\!B_{\rm c} \!+ \! \log_{2} \!\left(\frac{\psi_{\mathrm{vol}}}{\xi^{B_{\rm c}}} \!\right)\!\right\rceil
        \! \Rightarrow \!
        N_{\rm tr} \!\! = \!\! N_{\mathrm{c}} \!+\! M \!
        \!\leqslant\!\left\lceil\! B \! +\!\log_{2}\!\left(\frac{\psi_{\mathrm{vol}}}{\xi^{2NK}}\right)\right\rceil.
\end{equation*}
Since the number of candidate nodes at the $n$-th depth level is $2^{n}$, 
at most $\sum_{n=1}^{N_{\rm tr}}2^{n}=2^{N_{\rm tr}+1} -1 $ candidate nodes exist at the maximum tree depth $N_{\rm tr}$. 
In the worst case, all candidate nodes need to be fathomed. 
Thus, the algorithm terminates after performing at most $T=2^{N_{\rm tr}+1}-1\!\leqslant\! \left\lceil\frac{\psi_{\mathrm{vol}}}{\delta^{2NK}}2^{B+1} -1\right\rceil$ branching iterations, 
which ends the proof.

\end{document}